\documentclass[journal]{IEEEtran}

\usepackage[english]{babel}
\usepackage{booktabs}

\usepackage{ifpdf}

\usepackage{cite} 
\usepackage{url}
\usepackage{hyperref}
\usepackage[framemethod=default]{mdframed}

\ifCLASSINFOpdf
	\usepackage[pdftex]{graphicx}
	\graphicspath{{./figures/}}
\else
	\usepackage[dvips]{graphicx}
	\graphicspath{./figures/}
\fi
\usepackage{color}

\usepackage{pgf, tikz, pgfplots}
\usetikzlibrary{shapes, arrows, automata}
\usetikzlibrary{calc,hobby,decorations}

\usepackage[cmex10]{amsmath}
\usepackage{amsfonts, amssymb, amsthm}
\usepackage{mathrsfs}
\usepackage{tcolorbox}

\usepackage{algorithm} 
\usepackage{algorithmic}
\usepackage{enumerate}
\usepackage{multirow}
\usepackage{rotating}
\usepackage{subcaption}
\usepackage{needspace}

\input{mySymbol.sty}

\definecolor{penndarkestblue}{cmyk}{1,0.74,0,0.77}
\definecolor{penndarkerblue}{cmyk}{1,0.74,0,0.70}
\definecolor{pennblue}{cmyk}{0.99,0.66,0,0.57} 
\definecolor{pennlighterblue}{cmyk}{0.98,0.44,0,0.35}
\definecolor{pennlightestblue}{cmyk}{0.38,0.17,0,0.17} 

\definecolor{penndarkestred}{cmyk}{0,1,0.89,0.66}
\definecolor{penndarkerred}{cmyk}{0,1,0.88,0.55}
\definecolor{pennred}{cmyk}{0,1,0.83,0.42} 
\definecolor{pennlighterred}{cmyk}{0,1,0.6,0.24}
\definecolor{pennlightestred}{cmyk}{0,0.43,0.26,0.12} 

\definecolor{penndarkestgreen}{cmyk}{1,0,1,0.68}
\definecolor{penndarkergreen}{cmyk}{1,0,1,0.57}
\definecolor{penngreen}{cmyk}{1,0,1,0.44} 
\definecolor{pennlightergreen}{cmyk}{1,0,1,0.25}
\definecolor{pennlightestgreen}{cmyk}{0.43,0,0.43,0.13}

\definecolor{penndarkestorange}{cmyk}{0,0.65,1,0.49}
\definecolor{penndarkerorange}{cmyk}{0,0.65,1,0.33}
\definecolor{pennorange}{cmyk}{0,0.54,1,0.24} 
\definecolor{pennlighterorange}{cmyk}{0,0.32,1,0.13}
\definecolor{pennlightestorange}{cmyk}{0,0.15,0.46,0.06}
	
\definecolor{penndarkestpurple}{cmyk}{0,1,0.11,0.86}
\definecolor{penndarkerpurple}{cmyk}{0,1,0.13,0.82}
\definecolor{pennpurple}{cmyk}{0,1,0.11,0.71} 
\definecolor{pennlighterpurple}{cmyk}{0,1,0.05,0.46}
\definecolor{pennlightestpurple}{cmyk}{0,0.35,0.02,0.23}
	
\definecolor{pennyellow}{cmyk}{0,0.20,1,0.05} 
\definecolor{pennlightgray1}{cmyk}{0,0,0,0.05}
\definecolor{pennlightgray3}{cmyk}{0.01,0.01,0,0.18}
\definecolor{pennmediumgray1}{cmyk}{0.04,0.03,0,0.31}
\definecolor{pennmediumgray4}{cmyk}{0.08,0.06,0,0.54}
\definecolor{penndarkgray2}{cmyk}{0.09,0.07,0,0.71}
\definecolor{penndarkgray4}{cmyk}{0.1,0.1,0,0.92}

\def\SO3{\mathrm{SO(3)}}

\newtheorem{assumption}{\hspace{0pt}\bf Assumption \hspace{-0.15cm}}

\newtheorem{theorem}{\hspace{0pt}\bf Theorem}

\newtheorem{remark}{\hspace{0pt}\bf Remark}

\newtheorem{definition}{\hspace{0pt}\bf Definition}

\newenvironment{defBox}[1]
{
	\begin{mdframed}[hidealllines=false,backgroundcolor=gray!10]\textbf{#1.}
	}
	{
	\end{mdframed}
}

\begin{document}

\title{On the Trade-Off between Stability and Representational\! Capacity\! in\! Graph\! Neural\! Networks}

\author{\IEEEauthorblockN{Zhan Gao$^{ \dagger}$, Amanda Prorok$^{ \dagger}$ and Elvin Isufi$^{\ddagger }$}\\

%
\thanks{$^{\dagger}$Department of Computer Science and Technology, University of Cambridge, Cambridge, UK (Email: $\{$zg292, asp45$\}$@cam.ac.uk). $^{\ddagger }$Department of Intelligent Systems, Delft University of Technology, Delft, The Netherlands (Email: e.isufi-1@tudelft.nl). The work of Z. Gao and A. Prorok is supported by European Research Council (ERC) Project 949940 (gAIa). The work of E. Isufi is supported by the TU Delft AI Labs programme.}
}

\maketitle

\begin{abstract}
Analyzing the stability of graph neural networks (GNNs) under topological perturbations is key to understanding their transferability and the role of each architecture component. However, stability has been investigated only for particular architectures, questioning whether it holds for a broader spectrum of GNNs or only for a few instances. To answer this question, we study the stability of EdgeNet: a general GNN framework that unifies more than twenty solutions including the convolutional and attention-based classes, as well as graph isomorphism networks and hybrid architectures. We prove that all GNNs within the EdgeNet framework are stable to topological perturbations. By studying the effect of different EdgeNet categories on the stability, we show that GNNs with fewer degrees of freedom in their parameter space, linked to a lower representational capacity, are more stable. The key factor yielding this trade-off is the eigenvector misalignment between the EdgeNet parameter matrices and the graph shift operator. For example, graph convolutional neural networks that assign a single scalar per signal shift (hence, with a perfect alignment) are more stable than the more involved node or edge-varying counterparts. Extensive numerical results corroborate our theoretical findings and highlight the role of different architecture components in the trade-off. 
\end{abstract}

\begin{IEEEkeywords}
Graph neural networks, graph filters, stability property, topological perturbations
\end{IEEEkeywords}

\IEEEpeerreviewmaketitle

\section{Introduction}

Graph neural networks (GNNs) leverage the underlying data topology as an inductive bias to learn task-relevant representations in networked systems \cite{henaff2015deep, defferrard2016convolutional, gama2018convolutional, xupowerful, gao2022wide, wu2022graph, 10239277} and have found wide applications in multi-agent coordination \cite{gao2023Online, tolstaya2020learning, gao2023environment, 10251921}, recommendation systems \cite{fan2019graph, wu2019session} and wireless communications \cite{gao2020resource, chowdhury2021unfolding, gao2023decentralized1}. However, the topology on which GNNs are trained often changes during inference because of adversarial attacks \cite{jin2021adversarial}, estimation errors \cite{mei2015signal}, communication losses \cite{xu2004wireless} and mobility \cite{li2020graph}. This creates a mismatch between training and testing setups, ultimately, leading to degraded and unforeseeable behavior. Characterizing the stability of GNNs to topological perturbations is paramount as it sheds light on the transferability of the learned model, identifies the effects of architectural factors, and provides a handle to control the degradation \cite{gao2021training}. 

Stability of GNNs has attracted an increasing attention in the community. The works in \cite{gama2019diffusion, zou2020graph} studied the stability of non-trainable GNNs with graph wavelet filters to eigenvalue / eigenvector perturbations of the underlying graph under the diffusion distance \cite{coifman2006diffusion}. Subsequently, \cite{levie2021transferability} characterized the stability of graph convolutional neural networks (GCNNs) and established that GCNNs model similar representations on graphs that describe the same phenomenon. The work in \cite{gama2020stability} focused on relative perturbations, and showed that GCNNs are stable and discriminative at high graph frequencies, while \cite{kenlay2021stability, kenlay2021interpretable, wang2022graph} investigated the stability of GCNNs to structural perturbations with edge deletion or addition. The works in \cite{gao2021stochastic, gao2021stability} extended the stability results to stochastic perturbations using random edge sampling, while \cite{keriven2020convergence, ruiz2021graphon, maskey2023transferability} analyzed the stability with large random graphs based on a graphon analysis. Such findings have been used to improve the robustness of GNNs via spectral regularization \cite{9054341, 10051821} and constrained learning \cite{arghal2022robust, cervino2022training, gao2023learningSto}.

Stability of GNNs has also been studied under an adversarial attack perspective. The work in \cite{zugner2018adversarial, blumenkamp2021emergence} designed adversarial attacks on simple GCNNs with first-order filters \cite{kipfsemi}, to change edge and signal values for misclassifying the target node label, whereas \cite{dai2018adversarial} used reinforcement learning to generate adversarial attacks in node and graph classification tasks. The works in \cite{bojchevski2019certifiable, zugner2019certifiable} analyzed the stability of GCNNs under adversarial attacks in node classification problems, which focused on a subset of graph edges and showed that nodes will not change the learned label under these attacks. The work in \cite{ma2021graph} proposed an edge rewiring operation to perform adversarial attacks, and \cite{chang2021not} studied the vulnerability against adversarial attacks in the graph spectrum; the latter identified a certain low-frequency interval that is more robust to structural perturbations. 

The above works study the stability for particular GNN architectures, which leaves unanswered two questions: 

\smallskip
(i) \emph{Are general GNN families stable as a whole or is stability a property only of particular architectures?} 

\smallskip
(ii) \emph{What are the implications of achieving stability across the different GNN architectures?} 

\smallskip
\noindent Answering the first question is key to showing whether we can achieve stability across a broad spectrum of GNN architectures without restrictions to specific tasks, as well as to identifying the advantages and limitations of each particular architecture. Instead, answering the second question allows comparing the different GNN architectures not only in terms of representational capacity but also in terms of stability, ultimately, providing insights about the architecture selection for real-world applications.

To answer these questions, we consider the family of edge varying GNNs (EdgeNets), which is a general framework that unifies several GNN solutions, including GCNNs, ChebNets, GATs, GINs, and NGNs, among others \cite{isufi2021edgenets}. EdgeNet is a multi-layer architecture, where each layer considers multi-hop neighborhood information, and the architecture parameters are specific to each node, neighbor, and signal shift. Different subclasses of EdgeNets impose different restrictions on the architecture parameters to satisfy particular properties such as permutation equivariance and inductive transference, leading to different GNNs. Hence, proving the stability of EdgeNets guarantees the stability of all its subclasses, for most of which such results do not exist. Moreover, characterizing the stability of a general framework allows facilitating comparisons among the different GNNs and identifying the consequence of increasing the architecture parameters. Targeting this setting, we make the following contributions:
\begin{enumerate}
	\item We study the stability of a general GNN framework, i.e., EdgeNet, and identify the effects of different architecture factors. We prove that the stability holds for all GNNs in the EdgeNet framework, and the established results apply uniformly for all graphs and signals. 
	\item We show that subclasses of EdgeNets with more restrictions on the architecture parameters, i.e., lower degrees of freedom in the parameter space, are more stable to topological perturbations (e.g., GCNNs), while subclasses with higher degrees of freedom are less stable (e.g., GATs). This analysis yields an explicit trade-off between representational capacity and stability among the different GNNs, thereby identifying the penalty we pay for enhancing the expressive power of GNNs. 
	\item We show that the key factor affecting the stability-representational capacity trade-off across the different GNN architectures is the eigenvector misalignment between the edge varying parameter matrices in EdgeNets and the underlying graph shift operator. The more these eigenvectors align, the stronger the stability. This reveals the inherent reason behind the stability changes among the different GNNs. 
\end{enumerate}

The rest of this paper is organized as follows. Section \ref{sec:EdgeNet} introduces the EdgeNet and its subclasses, and illustrates how the general framework represents a variety of GNN architectures. Section \ref{sec:StabilitySub} studies the stability of the EdgeNet subclasses and analyzes the effects of architectural factors. Section \ref{sec:stabilityEdgeGF} extends the stability analysis to the general EdgeNets and puts forth an explicit trade-off between the representational capacity and stability. Numerical experiments corroborate the theoretical findings on both synthetic and real data in Section \ref{sec:experiments}, and the conclusions are drawn in Section \ref{sec:conclusion}. All proofs are collected in the supplementary material.

\section{Edge Varying Graph Neural Networks}\label{sec:EdgeNet}

Consider a graph $\ccalG=(\ccalV, \ccalE)$ of $n$ nodes $\ccalV = \{ 1, \ldots, n \}$ and $m$ edges $\ccalE = \{ (i,j) \}$ with the graph shift operator $\bbS$ (e.g., Laplacian, adjacency matrix) such that $[\bbS]_{ij} \neq 0$ iff $(i,j)\in \ccalE$ or $i=j$. Data supported on $\ccalG$ is a vector $\bbx \in \mathbb{R}^n$, referred to as the graph signal, where the $i$th entry $[\bbx]_i$ is the datum of node $i$ \cite{ortega2018graph}. The graph neural network (GNN) is an information processing architecture that learns task-relevant representations from the tuple $(\bbS, \bbx)$ \cite{isufi2022graph}. 

The edge varying graph neural network (EdgeNet) is a general framework that unifies different GNN architectures. It comprises cascaded layers, where each layer is composed of edge varying graph filters (EdgeGFs) followed by a pointwise nonlinearity \cite{isufi2021edgenets}. More specifically, the EdgeGF is a linear mapping of graph signals $\bbH_{{\rm Edge}}(\bbx, \bbS): \mathbb{R}^n \to \mathbb{R}^n$ of the form 
\begin{align}\label{eq:EdgeNet0}
	\bbH_{{\rm Edge}}(\bbx, \bbS) 
     := \sum_{k=0}^K \bbPhi^{(k)}\bbS^k\bbx, 
\end{align}
where $\{\bbPhi^{(k)} \in \mathbb{R}^{n \times n}\}_{k=0}^K$ are the filter parameter matrices sharing the support of $\bbS + \bbI$. This filter diffuses $\bbx$ over $\bbS$ to collect multi-hop neighborhood information up to a radius $K$ and leverages edge-dependent weights $\bbPhi^{(k)}$ that give different importance to the information coming from the different neighbors.

In an EdgeNet, each layer has $F$ input features $\{\bbx_{\ell-1}^f\}^{F}_{f=1}$ for $\ell=1,...,L$. These inputs are processed by $F^2$ EdgeGFs $\{\bbH_{{\rm Edge},\ell}^{fg}(\cdot, \bbS)\}_{f,g=1}^F$, aggregated over the input index $g$, and passed through the nonlinearity $\sigma(\cdot)$ to generate the $F$ output features 
\begin{align}\label{eq:EdgeNet2}
	\bbx_\ell^f = \sigma\Big(\sum_{g=1}^{F} &\bbH_{{\rm Edge}, \ell}^{fg}(\bbx_{\ell-1}^g, \bbS) \Big)~\for~f=1,\ldots,F,\\
	\label{eq:EdgeNet3}\bbx_0 &= \bbx,~ \bbx_L = \bbPsi_{\rm Edge}\big(\bbx,\bbS,\ccalH\big),
\end{align}
where $\bbPsi_{\rm Edge}\big(\bbx,\bbS,\ccalH\big): \mathbb{R}^n \to \mathbb{R}^n$ represents the nonlinear mapping of the EdgeNet and the set $\ccalH$ groups all filter parameter matrices. 

The EdgeNet assigns different weights to different edges at different hops. It unifies several widely used GNN solutions by imposing different 
restrictions on the edge varying parameter matrices $\bbPhi^{(k)}$ [cf. \eqref{eq:EdgeNet0}]. This allows establishing a theoretical analysis for a broad spectrum of GNN architectures, i.e., results that hold for general EdgeNets apply to any GNNs in the framework. To categorize the different types of GNNs in the EdgeNet framework, we propose two subclasses of EdgeNets: (i) Shift-Invariant EdgeNets (SI-EdgeNets) and (ii) Eigenvector-Sharing EdgeNets (ES-EdgeNets). 

\subsection{Shift-Invariant EdgeNets}

Shift-Invariant EdgeNets (SI-EdgeNets) consist of EdgeGFs that are invariant to the graph shift operation, i.e., $\bbH_{\rm Edge}(\bbS\bbx, \bbS) = \bbS \bbH_{\rm Edge}(\bbx, \bbS)$. This property requires the eigenvectors of the edge weight matrices $\{\bbPhi^{(k)}\}_{k=0}^K$ to coincide with those of the underlying graph $\bbS$. Specifically, let $\bbS = \bbV \bbLambda \bbV^\top$ be the eigendecomposition with eigenvectors $\bbV = [\bbv_1,...,\bbv_n]$ and eigenvalues $\bbLambda = \text{diag} (\lambda_1,...,\lambda_n)$. The shift invariance is equivalent to requiring the parameter matrices $\{\bbPhi^{(k)}\}_{k=0}^K$ to be fixed-support matrices that are diagonalizable w.r.t. 
$\bbV$ \cite{Coutino2019}, i.e., 
\begin{align}\label{eq:fixedSupportSIEdgeGF}
	\bbPhi^{(k)} \!\!\in\! \Omega_\bbV^{\ccalA} \!=\! \{ \bbPhi\!: \bbPhi \!=\! \bbV \bbphi \bbV^{-1}\!, [{\rm vec}(\bbPhi)]_i \!=\! 0,\!~\text{for}~\!i \!\in\! \ccalA \}
\end{align}
for $k=0,1,...,K$. Here, $\ccalA$ is the index set defining the zero entries of $\bbS + \bbI$, ${\rm vec}(\bbPhi)$ vectorizes the matrix $\bbPhi$ to an $n^2$-dimensional vector, and $\bbphi = \text{diag} (\phi_1,...,\phi_n)$ is any diagonal matrix. While \eqref{eq:fixedSupportSIEdgeGF} does not hold in general, there exists a subset of EdgeGFs that satisfies this condition -- see Appendix \ref{appendix:shiftInvariance}. By defining the Shift-Invariant EdgeGF (SI-EdgeGF) as 
\begin{align}\label{eq:SIEVGF}
	\bbH_{\rm SI}(\bbx, \bbS) \!:=\! \sum_{k=0}^{K}\! \bbPhi^{(k)}\bbS^k\bbx,~\bbPhi^{(k)} \!\in\! \Omega_\bbV^{\ccalA}~\for~k =0,...,K,
\end{align}
we construct the SI-EdgeNet following \eqref{eq:EdgeNet2}-\eqref{eq:EdgeNet3}.

Graph convolutional neural networks (GCNNs) are in the subclass of SI-EdgeNets, which particularize the edge weight matrices to the scaled identity matrices $\{\bbPhi^{(k)} = h^{(k)}\bbI\}_{k=0}^{K}$ \cite{gama2018convolutional}. This yields the graph convolutional filters (GCFs) 
\begin{align}\label{eq:GCF}
	\bbH_{\rm SI, Conv}(\bbx,\bbS) = \sum_{k=0}^{K} h^{(k)} \bbS^{k} \bbx. 
\end{align}
GCNNs assign the same weight to graph edges and thus, restrict the space of the architecture parameters. On the other hand, GCNNs are transferable to unseen graphs because scalar parameters are independent of the graph size, and are permutation equivariant because the filter aggregation of GCFs are pointwise [cf. \eqref{eq:GCF}] \cite{gama2020stability}. The former property facilitates generalization, which allows applying the trained model on different graphs during testing. The latter property facilitates data efficiency during training, since learning how to process one instance of the signal is equivalent to learning how to process all permutated versions of the same signal. Other GNNs in the SI-EdgeNet subclass include GCNs \cite{kipfsemi}, ChebNets \cite{defferrard2016convolutional}, JKNets \cite{xu2018representation}, and GINs \cite{xupowerful}, among others \cite{isufi2021edgenets}. 

\begin{figure}[t]
	\centering
	\includegraphics[width=0.85\linewidth , height=0.45\linewidth, trim=20 0 20 20]{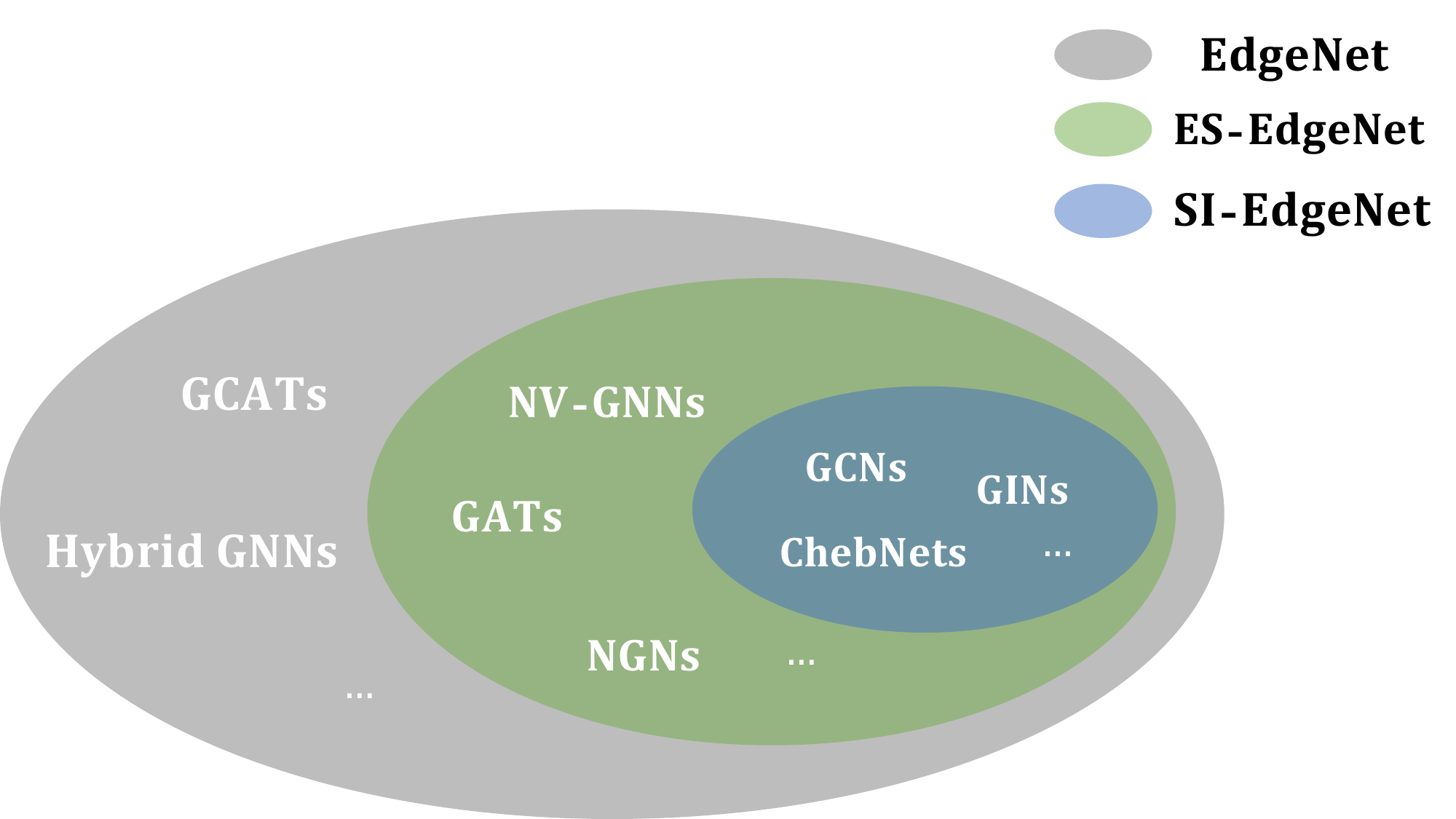}
	\caption{Relationship between EdgeNet [cf. \eqref{eq:EdgeNet0}], ES-EdgeNet [cf. \eqref{eq:EIEVGF}] and SI-EdgeNet [cf. \eqref{eq:SIEVGF}]. EdgeNet contains the subclass of ES-EdgeNet and ES-EdgeNet contains the subclass of SI-EdgeNet. The DoFs of the parameter space and the representational capacity increase from SI-EdeNet, ES-EdgeNet to the general EdgeNet given the same number of layers $L$ and per-layer features $F$.}
	\label{fig:set}
\end{figure}

\subsection{Eigenvector-Sharing EdgeNets}
Eigenvector-Sharing EdgeNets (ES-EdgeNets) require the edge weight matrices $\{\bbPhi^{(k)}\}_{k=0}^K$ to share the eigenvectors among them but they do not necessarily coincide with the eigenvectors of the graph shift operator $\bbS$. Specifically, let $\bbU$ be the eigenvectors of $\{\bbPhi^{(k)}\}_{k=0}^K$ that are different from the eigenvectors $\bbV$ of $\bbS$. Define also the set of fixed-support matrices that are diagonalizable w.r.t. $\bbU$ as [cf. \eqref{eq:fixedSupportSIEdgeGF}] 
\begin{align}\label{eq:fixedSupportESEdgeGFW}
	\Omega_\bbU^{\ccalA} \!=\! \{ \bbPhi: \bbPhi \!=\! \bbU \bbphi \bbU^{-1}, [{\rm vec}(\bbPhi)]_i \!=\! 0,~\text{for}~i \!\in\! \ccalA \},
\end{align}
and the union of the sets w.r.t. all possible eigenvectors $\bbU$ as
\begin{align}\label{eq:fixedSupportESEdgeGFW1}
	\Omega^{\ccalA} = \big\{ \cup_{\bbU} \Omega_\bbU^{\ccalA},~\forall~\bbU \big\}.
\end{align}
The Eigenvector-Sharing EdgeGF (ES-EdgeGF) is defined with the edge-weight matrices $\bbPhi^{(k)}$ selected from one of the fixed-support matrix set $\Omega^{\ccalA}_{\bbU}$ in $\Omega^{\ccalA}$, i.e.,
\begin{align}\label{eq:EIEVGF}
	\bbH_{\rm ES}(\bbx, \bbS) \!:=\!\! \sum_{k=0}^K\! \bbPhi^{(k)}\bbS^k\bbx,~\bbPhi^{(k)} \!\in\! \Omega_{\bbU}^{\ccalA}~\for~k\!=\!0,...,K,
\end{align}
which constructs the ES-EdgeNet following \eqref{eq:EdgeNet2}-\eqref{eq:EdgeNet3}. ES-EdgeNets extend the parameter space of SI-EdgeNets by allowing the eigenvectors of $\bbPhi^{(k)}$ to differ from those of $\bbS$. 

Node varying graph neural networks (NV-GNNs) are one instance of ES-EdgeNets, which particularize the edge weight matrices to the diagonal matrices $\{\bbPhi^{(k)} = \bbD^{(k)}\}_{k=0}^{K}$ \cite{gama2022node}. This yields the node varying graph filters (NV-GFs) 
\begin{align}\label{eq:NVGF}
	\bbH_{\rm ES, NV}(\bbx,\bbS) = \sum_{k=0}^{K} \bbD^{(k)} \bbS^{k} \bbx.
\end{align}
NV-GNNs aggregate neighborhood information in a way that they treat node features differently, which increases the number of architecture parameters. However, NV-GNNs are not permutation equivariant and are only transferable to graphs with the same number of nodes because parameters are node-wise. Other GNNs in the ES-EdgeNet subclass include GATs \cite{velivckovicgraph}, and NGNs \cite{de2020natural}, among others \cite{isufi2021edgenets}. 

In this work, we consider the degrees of freedom (DoFs) of the parameter space as a proxy for the representational capacity of GNNs. Thus, this representational capacity increases from SI-EdgeNets, ES-EdgeNets to the general EdgeNets -- see Fig. \ref{fig:set}. The goal is to investigate the stability of different GNNs in the framework of EdgeNets to graph perturbations, and establish connections between these GNN solutions in terms of stability and representational capacity.


\section{Stability for Subclasses of EdgeNets}\label{sec:StabilitySub}

The underlying graph topology $\bbS$ is often perturbed due to estimation errors, link losses, or adversarial attacks. GNNs will generate representations on the perturbed graph $\widetilde{\bbS}$ during testing, which deviate from the representations learned on the underlying graph $\bbS$ during training. The stability analysis characterizes the effect of graph perturbations on these representations. It provides insights to understanding the robustness of GNNs and sheds light on our handle to reduce the perturbation effects. We aim to establish stability results that hold for all GNNs in the EdgeNet framework, rather than for particular instances such as GCNNs in \cite{gama2020stability}. This also allows to show how stability changes across different GNNs, and to identify its inherent trade-off with the representational capacity. 

We study the effect of graph perturbations on GNNs w.r.t. the \textit{relative} perturbation model.
\begin{definition}[Relative Perturbation \cite{gama2020stability}]\label{def:relativePerturbation}
    Given the underlying graph $\bbS$, the set of relative perturbations is defined as
    \begin{align}\label{eq:relativePerturbation}
        \ccalE(\bbS) = \{ \bbE \in \mathbb{R}^{n \times n}: \widetilde{\bbS} = \bbS + \bbE \bbS + \bbS \bbE, \bbE = \bbE^\top \},
    \end{align}
    where $\widetilde{\bbS}$ is the perturbed graph. 
\end{definition}
\noindent The perturbation matrix $\bbE$ is symmetric and the perturbation size $\|\bbE\|_2$ measures the dissimilarity between the underlying graph $\bbS$ and the perturbed graph $\widetilde{\bbS}$. Such a model ties the perturbation size to the graph topology through the multiplication of $\bbE$ with $\bbS$ and thus, captures the structural information. For example, a small $\bbE$ may result in a large perturbation $\bbE \bbS + \bbS \bbE$ if the underlying graph $\bbS$ is dense with large edge weights. The latter is not possible if considering the absolute perturbation $\widetilde{\bbS} = \bbS + \bbE$ -- see \cite{gama2020stability}.\footnote{Characterizing stability for the relative perturbation model is also more challenging than the absolute one as discussed in \cite{gama2020stability}. Our results could be proven also for the absolute perturbation but we do not detail them here.}

\subsection{Stability of Shift-Invariant EdgeGFs}\label{subsec:stabilitySIGF}

We conduct the stability analysis in the graph spectral domain \cite{ortega2018graph} to establish results that hold uniformly for any graphs. In particular, the graph Fourier transform (GFT) expands the graph signal $\bbx$ over the eigenvectors $\bbV$ of $\bbS$ as $\hat{\bbx} = \sum_{i=1}^n \hat{x}_i \bbv_i$. By substituting the GFT into the SI-EdgeGF [cf. \eqref{eq:SIEVGF}] and using the fact that $\{\bbPhi^{(k)}\}_{k=0}^K$ share the same eigenvectors $\bbV$, the filter output is 
\begin{equation}\label{eq:FilterGFT}
	\bbH_{\rm SI}(\bbx, \bbS) = \bby = \sum_{i=1}^n \hat{x}_i \sum_{k=0}^K \phi_i^{(k)} \lambda_i^k \bbv_i, 
\end{equation}
where $\{\phi_i^{(k)}\}_{i=1}^n$ and $\{\lambda_i\}_{i=1}^n$ are the eigenvalues of $\bbPhi^{(k)}$ and $\bbS$. By applying the GFT on the output $\bby = \sum_{i=1}^n\hat{y}_i \bbv_i$, we have the spectral input-output relation $\hat{y}_i = \sum_{k=0}^K \phi_i^{(k)} \lambda_i^k \hat{x}_i$ for $i=1,\ldots,n$. That is, we have $n$ different spectral responses, each of which operates on a different graph frequency $\lambda_i$ with parameters determined by the eigenvalues 
$\{\phi_i^{(k)}\}_{k=0}^K$ of the edge weight matrices. This motivates to define the shift-invariant filter frequency response for the SI-EdgeGF. 
\begin{definition}[Shift-invariant filter frequency response]
	Consider an SI-EdgeGF \eqref{eq:SIEVGF} with the edge weight matrices $\{\bbPhi^{(k)}\}_{k=0}^K$. Let $\{\phi_i^{(k)}\}_{i=1}^n$ be the eigenvalues of $\bbPhi^{(k)}$ for $k=0,\ldots,K$. The shift-invariant filter frequency response is a collection of $n$ index-specific analytic functions of the form 
	\begin{align}\label{eq:SIfrequencyResponse}
		h_i(\lambda) := \sum_{k=0}^K \phi_i^{(k)} \lambda^k,~\for~i=1,\ldots,n,
	\end{align}
    where $\{\phi_i^{(k)}\}_{k=0}^K$ are function parameters that determine the function shape, and $\lambda$ is the function variable. 
\end{definition}
\noindent The edge weight matrices $\{\bbPhi^{(k)}\}_{k=0}^K$ instantiate the function parameters of $h_i(\lambda)$ with the eigenvalues $\{\phi_i^{(k)}\}_{k=0}^K$, while the underlying graph $\bbS$ specifies the frequency variable $\lambda$ of $h_i(\lambda)$ on the eigenvalues $\lambda_i$ for $i=1,...,n$ -- see Fig. \ref{subfigfb}. 

The perturbed graph $\widetilde{\bbS}$ will instantiate the frequency variable $\lambda$ on the perturbed eigenvalues $\{\widetilde{\lambda}_i\}_{i=1}^n$, which indicates that the variability of the frequency responses $\{h_i(\lambda)\}_{i=1}^n$ plays an important role in the filter's stability. We introduce the integral Lipschitz SI-EdgeGF to characterize this spectrum variability.  
\begin{definition}[Integral Lipschitz SI-EdgeGF]\label{def:LipschitzFilterFrequencyResponse}
	Consider an SI-EdgeGF with the shift-invariant filter frequency response $\{h_i(\lambda)\}_{i=1}^n$ [cf. \eqref{eq:SIfrequencyResponse}]. The filter is integral Lipschitz if there exists a constant $C_{L}>0$ such that for any $\lambda_1, \lambda_2$, it holds that
	\begin{equation}\label{eq:LipschitzFrequencyResponse}
		\Big|\frac{\lambda_1+\lambda_2}{2} \cdot \frac{h_i(\lambda_1) - h_i(\lambda_2)}{\lambda_1 - \lambda_2}\Big| \le C_{L},~\text{for}~i=1,...,n. 
	\end{equation}
\end{definition}
\noindent The integral Lipschitz SI-EdgeGF is equivalent to requiring that the derivative of the shift-invariant filter frequency response satisfies 
\begin{align}
	|\lambda h_i'(\lambda)| \le C_{L},~\text{for~all}~i=1,...,n.
\end{align}
It restricts the variability of the frequency responses in a way that $h_i(\lambda)$ may change rapidly in low frequencies, 
i.e., small $\lambda$, but tends to be flat in high frequencies, i.e., large $\lambda$. This condition is reminiscent of the scale invariance of wavelet transforms \cite{hammond2011wavelets, shuman2015spectrum, gama2020stability}. We now characterize the stability of SI-EdgeGFs.
\begin{theorem}[Stability of SI-EdgeGF]\label{theorem:SIEVfilterStability}
    Consider the SI-EdgeGF $\bbH_{\rm SI}(\bbx, \bbS)$ [cf. \eqref{eq:FilterGFT}]. Let $h_i(\lambda)$ be the shift-invariant filter frequency response that is integral Lipschitz w.r.t. $C_{L}$ [cf. \eqref{eq:LipschitzFrequencyResponse}] and satisfies $|h_i(\lambda)| \le 1$ for $i=1,\ldots,n$, and $\widetilde{\bbS}$ be the perturbed graph with perturbation size $\epsilon$ [Def. \ref{def:relativePerturbation}]. Then, for any signal $\bbx$, it holds that
	\begin{equation}\label{eq:SIEVFilterStability}
		\| \bbH_{\rm SI}(\bbx, \bbS) \!-\! \bbH_{\rm SI}(\bbx, \widetilde{\bbS}) \|_2 \!\le\! C_{\rm SI} \|\bbx\|_2 \eps \!+\! \ccalO\big(\eps^2\big),
	\end{equation}
	where $C_{\rm SI} = 2 \sqrt{n} C_L$ is the stability constant. 
\end{theorem}
\begin{proof}
    See Appendix \ref{proof:thm1}.
\end{proof}

\begin{figure*}%
\begin{tcolorbox}[width=\textwidth, colback = gray!10, colframe = gray, arc = 0pt, outer arc = 0pt, boxrule = 1pt]
\centering
\begin{subfigure}{0.25\columnwidth}
\includegraphics[width=1\linewidth, height = 0.66\linewidth, trim = {0cm 0cm 0cm 0cm}, clip]{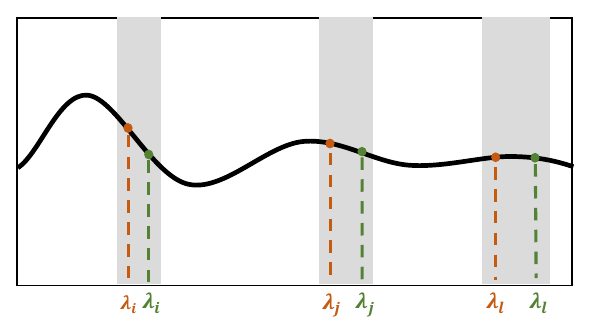}%
\caption{}%
\label{subfigfa}%
\end{subfigure}\hfill%
\begin{subfigure}{0.25\columnwidth}
\includegraphics[width=1\linewidth,height = 0.66\linewidth, trim = {0cm 0cm 0cm 0cm}, clip]{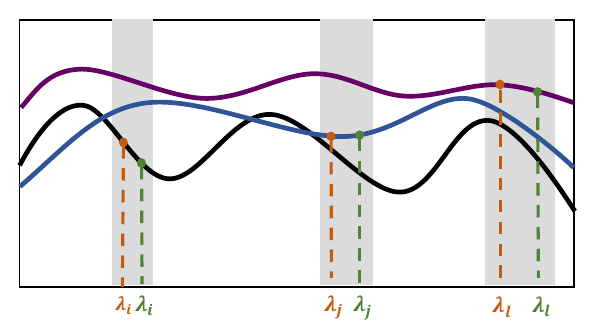}%
\caption{}%
\label{subfigfb}%
\end{subfigure}\hfill%
\begin{subfigure}{0.25\columnwidth}
\includegraphics[width=1\linewidth,height = 0.66\linewidth, trim = {0cm 0cm 0cm 0cm}, clip]{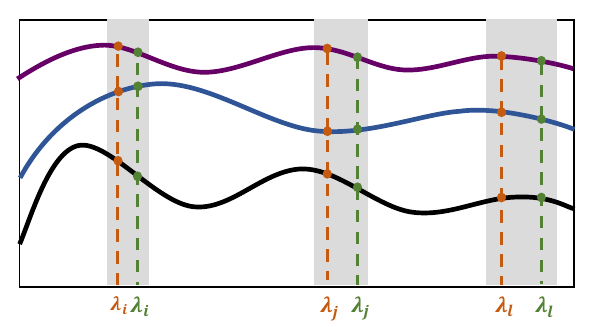}%
\caption{}%
\label{subfigfc}%
\end{subfigure}\hfill%
\begin{subfigure}{0.25\columnwidth}
\includegraphics[width=1\linewidth,height = 0.66\linewidth, trim = {2.5cm 1cm 1.5cm 2cm}, clip]{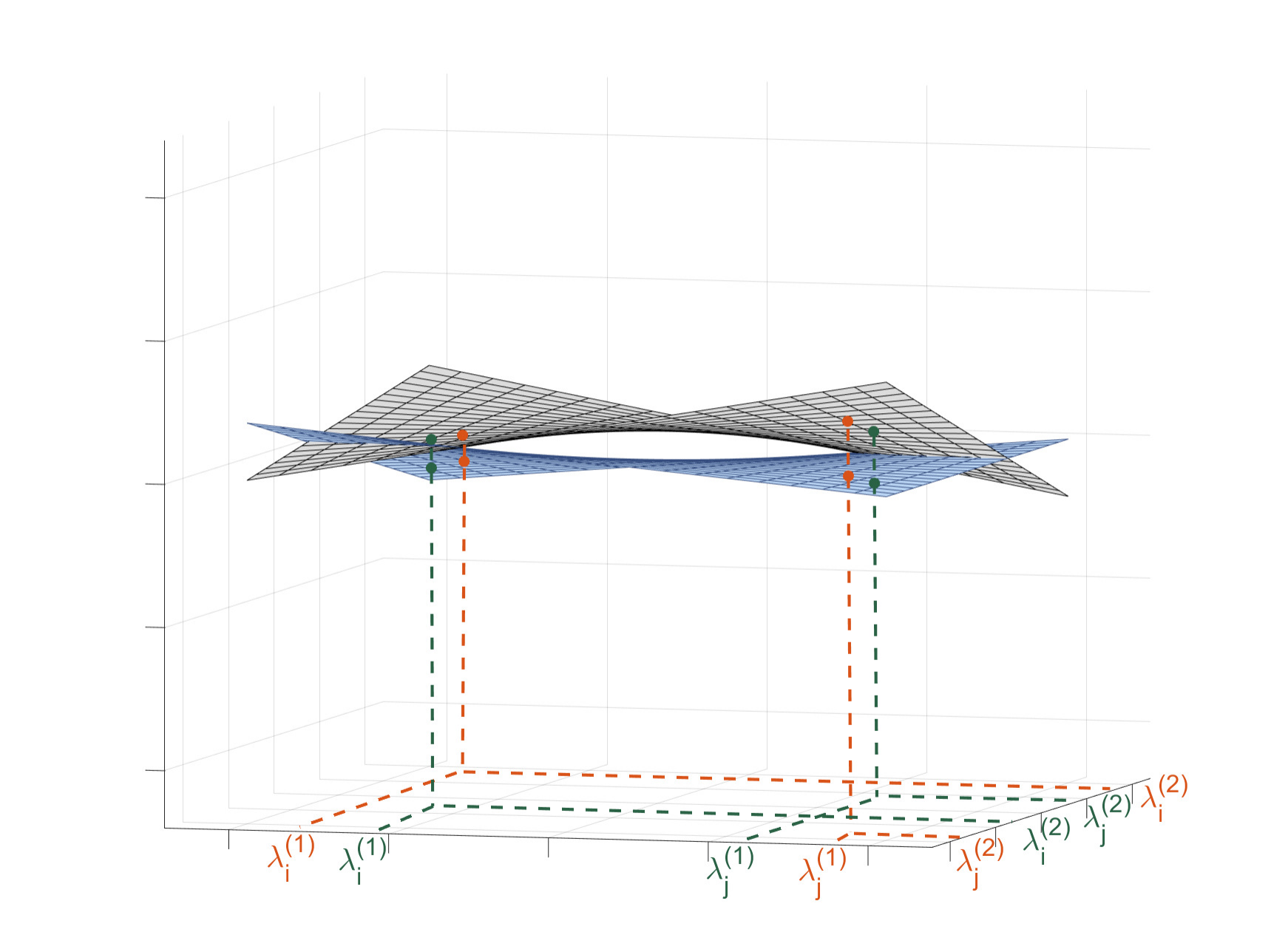}%
\caption{}%
\label{subfigfd}%
\end{subfigure}
\caption{Graph perturbations on the filter frequency responses of different EdgeNet categories. In all panels, we highlight two different graph instantiations in red (e.g., true graph eigenvalues) and green (e.g., perturbed graph eigenvalues). \textbf{(a)} Filter frequency response of integral Lipschitz graph convolutional filter (GCF) $h(\lambda) = \sum_{k=0}^K h_k \lambda^k$ \cite{gama2018convolutional}. The function $h(\lambda)$ is determined by the filter coefficients $\{h_k\}_{k=0}^K$ and the function variable $\lambda$ is specified by the eigenvalues of the graph $\{\lambda_i\}_{i=1}^n$. For a given graph, the integral Lipschitz property requires $h(\lambda)$ to satisfy \eqref{eq:LipschitzFrequencyResponse} w.r.t. all $n$ frequencies $\{\lambda_i\}_{i=1}^n$ -- see Remark \ref{remark:conditionFlexibility}. \textbf{(b)} Three example filter frequency responses of integral Lipschitz SI-EdgeGF $h_i(\lambda)$ (black line), $h_j(\lambda)$ (blue line) and $h_\ell(\lambda)$ (purple line) [cf. \eqref{eq:SIfrequencyResponse}]. The function $h_i(\lambda) \setminus h_j(\lambda) \setminus h_\ell(\lambda)$ is determined by the eigenvalues of the edge weight matrices $\{\phi_i^{(k)}\}_{k=0}^K \setminus \{\phi_j^{(k)}\}_{k=0}^K \setminus \{\phi_\ell^{(k)}\}_{k=0}^K$. The function variable $\lambda$ is specified by the eigenvalues of the graph $\{\lambda_i\}_{i=1}^n$ and each of them evaluates one response. For a given graph, the integral Lipschitz property only requires each $h_i(\lambda) \setminus h_j(\lambda) \setminus h_\ell(\lambda)$ to satisfy \eqref{eq:LipschitzFrequencyResponse} w.r.t. a single frequency $\lambda_i \setminus \lambda_j \setminus \lambda_\ell$, which provides more flexibility for the design of the filter frequency responses -- see Remark \ref{remark:conditionFlexibility}. \textbf{(c)} Three example filter frequency responses of integral Lipschitz ES-EdgeGF $h_i(\lambda)$ (black line), $h_j(\lambda)$ (blue line) and $h_\ell(\lambda)$ (purple line) [cf. \eqref{eq:ESFrequencyResponse}]. The function $h_i(\lambda) \setminus h_j(\lambda) \setminus h_\ell(\lambda)$ is determined by the eigenvalues of the edge weight matrices $\{\phi_i^{(k)}\}_{k=0}^K \setminus \{\phi_j^{(k)}\}_{k=0}^K \setminus \{\phi_\ell^{(k)}\}_{k=0}^K$ and the function variable $\lambda$ is specified by the eigenvalues of the graph $\{\lambda_i\}_{i=1}^n$. For a given graph, the integral Lipschitz property requires each $h_i(\lambda) \setminus h_j(\lambda) \setminus h_\ell(\lambda)$ to satisfy \eqref{eq:LipschitzFrequencyResponse} w.r.t. the $n$ frequencies $\{\lambda_i\}_{i=1}^n$ -- see Remark \ref{remark:ESEdgeGFFrequency}. \textbf{(d)} Two example filter frequency responses of integral Lipschitz EdgeGFs $h_i(\bblambda)$ (black surface) and $h_j(\bblambda)$ (blue surface) with $K=2$. The function $h_i(\bblambda) \setminus h_j(\bblambda)$ is determined by the eigenvalues of the edge weight matrices $\{\phi^{(k)}_i\}_{k=0}^K \setminus \{\phi^{(k)}_j\}_{k=0}^K$, and the multivariate function variable $\bblambda$ is specified by the eigenvalues of the graph $\{\lambda_i\}_{i=1}^n$ and the eigenvector misalignment between the graph $\bbS$ and the edge weight matrices  $\{\bbPhi^{(k)}_i\}_{k=0}^K$ [cf. \eqref{eq:GeneralizedFrequencyResponse45}-\eqref{eq:GeneralizedFrequencyResponse5}]. For a given graph, the integral Lipschitz property requires each $h_i(\lambda) \setminus h_j(\lambda)$ satisfying \eqref{eq:LipschitzGeneralizedResponse} w.r.t. $n$ multivariate frequencies $\{\bblambda_i\}_{i=1}^n$ -- see Remark \ref{remark:EdgeGFFrequency}.}\label{fig:filterFrequencyResponses}
\end{tcolorbox}
\end{figure*} 

Theorem \ref{theorem:SIEVfilterStability} states that the output difference of the SI-EdgeGF over the underlying graph $\bbS$ and the perturbed graph $\widetilde{\bbS}$ is bounded linearly by the perturbation size $\eps$ w.r.t. a stability constant $C_{\rm SI}$. When perturbations are small (i.e., $\eps \to 0$), the stability bound approaches zero, and the SI-EdgeGF maintains the performance. The stability constant $C_{\rm SI}$ embeds the role of the filter and graph: $C_{L}$ is the Lipschitz constant of filter frequency responses $\{h_i(\lambda)\}_{i=1}^n$ and $n$ is the graph size. The Lipschitz property of $\{h_i(\lambda)\}_{i=1}^n$ is determined by the eigenvalues of the edge weight matrices $\bbPhi^{(k)}$ but not by the specific edge weights. The stability can be improved by designing the edge weight matrices $\bbPhi^{(k)}$ to reduce the Lipschitz constant $C_{L}$, which allows less variability of $\{h_i(\lambda)\}_{i=1}^n$ between nearby frequencies and is more stable to frequency deviations caused by the graph perturbations. The bound in \eqref{eq:SIEVFilterStability} is however the worst-case analysis of the output deviation, i.e., the Lipschitz constant $C_L$ is the maximum upper bound of all frequency responses $\{h_i(\lambda)\}_{i=1}^n$ and the Lipschitz property holds uniformly for all graph signals $\bbx$. 

By particularizing the edge weight matrices $\{\bbPhi^{(k)}\}_{k=0}^K$ to scaled identity matrices $\{h^{(k)}\bbI\}_{k=0}^K$, all frequency responses $\{h_i(\lambda)\}_{i=1}^n$ collapse to a single frequency response as that of the GCF [cf. \eqref{eq:GCF}], and the integral Lipschitz SI-EdgeGF reduces to the integral Lipschitz GCF. Theorem \ref{theorem:SIEVfilterStability} recovers the stability of GCFs, i.e., the stability bound in \eqref{eq:SIEVFilterStability} specifies that of GCFs. 

\begin{remark}\label{remark:conditionFlexibility}
	Theorem \ref{theorem:SIEVfilterStability} holds for any graph $\bbS$. However, if $\bbS$ is fixed or given by the specific problem, we can 
	improve the stability by designing graph-specific frequency responses $\{h_i(\lambda)\}_{i=1}^n$ to reduce the Lipschitz constant $C_L$. In particular, each frequency response $h_i(\lambda)$ only instantiates on a single frequency $\lambda_i$ of $\bbS$ for $i=1,\ldots,n$, such that the condition \eqref{eq:LipschitzFrequencyResponse} only needs to be satisfied w.r.t. a single value $\lambda_i$,\footnote{The perturbed graph $\widetilde{\bbS}$ may have different eigenvalues $\{\widetilde{\lambda}_i\}_{i=1}^n$ from those of the underlying graph $\bbS$. However, our stability result does not require the integral Lipschitz condition \eqref{eq:LipschitzFrequencyResponse} to be satisfied w.r.t. perturbed eigenvalues $\{\widetilde{\lambda}_i\}_{i=1}^n$ but only w.r.t. original eigenvalues $\{\lambda_i\}_{i=1}^n$, which demonstrates the strength of the stability analysis developed in this paper.} i.e.,
	\begin{equation}\label{eq:LipschitzFrequencyResponseSpecific}
	\Big|\frac{\lambda+\lambda_i}{2} \cdot \frac{h_i(\lambda) - h_i(\lambda_i)}{\lambda - \lambda_i}\Big| \le C_{L}, 
	\end{equation}
    where $\lambda \in \{\lambda_1,\ldots,\lambda_n\}$ is another frequency of $\bbS$. This provides more flexibility for the design of $h_i(\lambda)$, relaxes the integral Lipschitz condition \eqref{eq:LipschitzFrequencyResponse}, and allows a lower 
	$C_L$ to improve stability. It is worth noting that since $\{h_i(\lambda)\}_{i=1}^n$ collapse to a single $h(\lambda)$ for GCFs [cf. \eqref{eq:GCF}], this $h(\lambda)$ needs to satisfy the Lipschitz condition \eqref{eq:LipschitzFrequencyResponse} w.r.t. all $n$ frequencies $\{\lambda_i\}_{i=1}^n$, which decreases the flexibility of the filter design compared with SI-EdgeGFs. 	
\end{remark}

\subsection{Stability of Eigenvector-Sharing EdgeGFs}\label{subsec:stabilityESGF}

In ES-EdgeGFs the edge weight matrices $\{\bbPhi^{(k)}\}_{k=0}^K$ share the eigenvectors, but they may be different from the eigenvectors $\bbV$ of $\bbS$. Specifically, let $\{\bbPhi^{(k)} = \bbU \bbphi^{(k)} \bbU^\top\}_{k=0}^K$ be the eigendecompositions where $\{\bbphi^{(k)}\}_{k=0}^K = \{\text{diag}(\phi^{(k)}_1,\ldots,\phi^{(k)}_n)\}_{k=0}^{K}$ are the eigenvalues and $\bbU$ are the eigenvectors. We characterize the alignment between eigenvectors $\bbV$ and $\bbU$ as follows. 
\begin{definition}[$\varepsilon$-misalignment of eigenvectors]\label{def:misalignment}
	Let $\bbV = [\bbv_1,\ldots,\bbv_n]$ and $\bbU = [\bbu_1,\ldots,\bbu_n]$ be two sets of eigenvectors, $\bbv_i = \sum_{j=1}^n \hat{v}_{ij} \bbu_j$ be the expansion of $\bbv_i$ over $\bbU$, and $\bbu_j = \sum_{i=1}^n \hat{u}_{ji} \bbv_i$ be the expansion of $\bbu_j$ over $\bbV$ for $i,j=1,\ldots, n$. Then, $\bbV$ and $\bbU$ are $\varepsilon$-misaligned if 
	\begin{align}
		&|\hat{u}_{ij}| \le \varepsilon~for~i=1,\ldots,n~\text{and}~j \ne i,\\
		&|\hat{v}_{ji}| \le \varepsilon~for~j=1,\ldots,n~\text{and}~i \ne j.
	\end{align} 
\end{definition} \vspace{-1mm}
\noindent When $\bbV$ is close to $\bbU$, $\bbv_i$ is close to $\bbu_i$, the diagonal coefficients $\{\hat{u}_{ii}\}_{i=1}^n$ and $\{\hat{v}_{jj}\}_{j=1}^n$ are large, and the off-diagonal coefficients $\{\hat{u}_{ij}\}_{j\ne i}$ and $\{\hat{v}_{ji}\}_{i\ne j}$ are small. This indicates that $\bbV$ is more aligned with $\bbU$ if $\varepsilon$ is smaller. For the exact alignment $\bbU = \bbV$, we have $\{\hat{u}_{ii} = 1\}_{i=1}^n$, $\{\hat{v}_{jj} = 1\}_{j=1}^n$ and $\{\hat{u}_{ij}=0\}_{j\ne i}$, $\{\hat{v}_{ji} = 0\}_{i\ne j}$; hence, $\bbV$ and $\bbU$ are $\varepsilon$-misaligned w.r.t. $\varepsilon=0$. 

We define the filter frequency responses of the ES-EdgeGF in a similar manner. By expanding the signal $\bbx$ over the graph shift operator's eigenvectors $\bbV$ as $\bbx = \sum_{i=1}^n \hat{x}_i \bbv_i$, the eigenvector $\bbv_i$ over the edge weight matrices' eigenvectors $\bbU$ as $\bbv_i = \sum_{j=1}^n \hat{v}_{ij} \bbu_j$, and $\bbu_j$ over $\bbV$ as $\bbu_j = \sum_{\ell=1}^n \hat{u}_{j\ell} \bbv_\ell$, the filter output reads as 
\begin{align}\label{eq:spectralESEdgeGF0}
	\bbH_{\rm ES}(\bbx,\bbS) = \sum_{i=1}^n \sum_{j=1}^n \sum_{\ell=1}^n \hat{x}_i \sum_{k=0}^K\phi_j^{(k)} \hat{v}_{ij} \hat{u}_{j \ell} \lambda_{i}^k \bbv_\ell.
\end{align} 
The expression in \eqref{eq:spectralESEdgeGF0} is akin to \eqref{eq:FilterGFT}, while the misalignment between eigenvectors $\bbV$ and $\bbU$ scales all graph frequency terms by the same factor $\hat{v}_{ij} \hat{u}_{j \ell}$. It allows to extract this common factor and re-write \eqref{eq:spectralESEdgeGF0} as
\begin{align}\label{eq:spectralESEdgeGF}
	\bbH_{\rm ES}(\bbx,\bbS) = \sum_{i=1}^n \sum_{j=1}^n \sum_{\ell=1}^n \hat{x}_i \hat{v}_{ij} \hat{u}_{j \ell} \sum_{k=0}^K\phi_j^{(k)} \lambda_{i}^k \bbv_\ell.
\end{align} 
In turn, this motivates to define the eigenvector-sharing filter frequency response of the ES-EdgeGF as
\begin{align}\label{eq:ESFrequencyResponse}
	h_j(\lambda) = \sum_{j=1}^n \phi_j^{(k)} \lambda,~\text{for}~j=1,\ldots,n.
\end{align}
While \eqref{eq:ESFrequencyResponse} takes the same form as that of the SI-EdgeGF [cf. \eqref{eq:SIfrequencyResponse}], each eigenvector-sharing filter frequency response $h_j(\lambda)$ is evaluated at all $n$ frequencies $\{\lambda_i\}_{i=1}^n$ because of the eigenvector misalignment [cf. \eqref{eq:spectralESEdgeGF}], instead of a single frequency $\lambda_j$ [cf. \eqref{eq:FilterGFT}]. We then define the integral Lipschitz ES-EdgeGF as in Definition \ref{def:LipschitzFilterFrequencyResponse} -- see Fig. \ref{subfigfc}, and claim the following theorem. 
\begin{theorem}[Stability of ES-EdgeGF]\label{theorem:EI-EVfilterStability}
	Consider the ES-EdgeGF $\bbH_{\rm ES}(\bbx, \bbS)$ with the edge weight matrices $\{\bbPhi^{(k)}\}_{k=0}^{K}$ [cf. \eqref{eq:EIEVGF}]. Let $h_i(\lambda)$ be the eigenvector-sharing filter frequency response that is integral Lipschitz w.r.t. $C_L$ and satisfies $|h_i(\lambda)| \le 1$ for $i=1,...,n$, $\bbV$ and $\bbU$ be the eigenvectors of $\bbS$ and $\{\bbPhi^{(k)}\}_{k=0}^{K}$ that are $\varepsilon$-misaligned [Def. \ref{def:misalignment}], and $\widetilde{\bbS}$ be the perturbed graph with perturbation size $\eps$ [Def. \ref{def:relativePerturbation}]. Then, for any graph signal $\bbx$, it holds that 
	\begin{equation}\label{eq:EI-EVFilterStability}
		\| \bbH_{\rm ES}(\bbx, \bbS)\!-\! \bbH_{\rm ES}(\bbx, \widetilde{\bbS}) \|_2 \!\le\! C_{\rm ES} \| \bbx \|_2 \eps \!+\! \ccalO\big(\eps^2\big),
	\end{equation}
	where $C_{\rm ES} = 2\sqrt{n}(1+n\varepsilon) C_L$ is the stability constant.
\end{theorem} 
\begin{proof}
    See Appendix \ref{proof:theorem2}.
\end{proof}

Theorem \ref{theorem:EI-EVfilterStability} states that ES-EdgeGFs are Lipschitz stable to graph perturbations, but with an increased stability constant $C_{\rm ES}$ compared with SI-EdgeGFs. Such a constant comprises three terms: (i) the Lipschitz constant $C_L$; (ii) the graph size $n$; and (iii) the eigenvector misalignment $\varepsilon$. The first two terms are inherited from SI-EdgeGFs [Thm. \ref{theorem:SIEVfilterStability}]. The third term is the consequence of graph perturbations propagating through the different eigenbases of the underlying graph $\bbS$ and edge weight matrices $\{\bbPhi^{(k)}\}_{k=0}^{K}$. This amplifies the frequency deviations caused by the graph perturbations and increases the perturbation effect on the output. The more aligned the eigenbases are, the smaller $\varepsilon$, and the more stable the ES-EdgeGF. Since $\varepsilon$ is smaller than one [Def. \ref{def:misalignment}], the stability constant $C_{\rm ES}$ is upper bounded by $2\sqrt{n}(1+n) C_L$, which is the worst-case analysis. The result demonstrates that while the increased flexibility of the edge weight matrices of ES-EdgeGFs expands the parameter space and improve the representational capacity, the filter loses in stability, yielding an explicit trade-off between two factors. 

If the eigenvectors $\bbU$ are aligned with $\bbV$, ES-EdgeGFs reduce to SI-EdgeGFs and the stability bound in \eqref{eq:EI-EVFilterStability} reduces to that in \eqref{eq:SIEVFilterStability}. The stability of SI-EdgeNets / ES-EdgeNets inherits from that of SI-EdgeGFs / ES-EdgeGFs -- see Theorem \ref{thm:stabilityEdgeNet} in Section \ref{sec:stabilityEdgeGF}.

\begin{remark}\label{remark:ESEdgeGFFrequency}
	For a fixed or given graph $\bbS$, integral Lipschitz ES-EdgeGFs require more restrictions on $\{h_j(\lambda)\}_{j=1}^n$ than SI-EdgeGFs [Remark \ref{remark:conditionFlexibility}]. Specifically, ES-EdgeGFs instantiate each $h_j(\lambda)$ on $n$ frequencies $\{\lambda_i\}_{i=1}^n$ because of the eigenvector misalignment between $\{\bbPhi^{(k)}\}_{k=0}^K$ and $\bbS$ [cf. \eqref{eq:spectralESEdgeGF}]. This requires each $h_j(\lambda)$ satisfying the condition \eqref{eq:LipschitzFrequencyResponse} w.r.t. all $\{\lambda_i\}_{i=1}^n$, while SI-EdgeGFs only require each $h_j(\lambda)$ satisfying \eqref{eq:LipschitzFrequencyResponse} w.r.t. a single $\lambda_j$. 
\end{remark}

\section{Stability For General EdgeNets}\label{sec:stabilityEdgeGF}

In this section, we discuss the stability of the general EdgeNets [cf. \eqref{eq:EdgeNet0}] with no restriction on the edge varying parameter matrices $\{\bbPhi^{(k)}\}_{k=0}^K$, which lead to the strongest representational capacity. 

\subsection{Edge Varying Filter Frequency Response}

To conduct the graph-universal stability analysis in this general case, we first generalize the concept of the filter frequency response. Specifically, the EdgeGF contains edge weight matrices $\{\bbPhi^{(k)}\}_{k=0}^K$ with different eigenvectors that cannot be handled by the conventional spectral approach. The following example illustrates this. 

\begin{defBox}{Example}
	Consider an EdgeGF of order $K=2$, i.e.,
	\begin{align}\label{eq:exampleFilter}
		\bbH_{{\rm Edge}}(\bbx, \bbS) &:= \bbPhi^{(0)}\bbx + \bbPhi^{(1)}\bbS\bbx + \bbPhi^{(2)}\bbS^2\bbx. 
	\end{align}
	Consider also the eigendecomposition $\bbPhi^{(k)} = \bbU^{(k)} \bbphi^{(k)} {\bbU^{(k)}}^\top$ with eigenvectors $\bbU^{(k)} = [\bbu_1^{(k)},\ldots,\bbu_n^{(k)}]$ and eigenvalues $\bbphi^{(k)} = \text{diag}(\phi_{1}^{(k)}, \ldots, \phi_{n}^{(k)})$ for $k=0, 1, 2$. We first expand the graph signal $\bbx$ over the eigenvectors of the graph shift operator $\bbV$ as $\bbx = \sum_{i=1}^n \hat{x}_{i} \bbv_i$, and then the eigenvector $\bbv_i$ over $\bbU^{(0)}$ as $\bbv_i = \sum_{j=1}^n \hat{v}_{ij}^{(0)} \bbu_j^{(0)}$. Likewise, we can also expand $\bbu_j^{(0)}$ over $\bbV$ as $\bbu_j^{(0)} = \sum_{\ell=1}^n \hat{u}_{j\ell}^{(0)} \bbv_\ell$. By substituting these expansions into the first term $\bbPhi^{(0)} \bbx$ in \eqref{eq:exampleFilter}, we have 
	\begin{align}\label{eq:GeneralizedFrequencyResponse1}
		\bbPhi^{(0)} \bbx & = \sum_{i=1}^n \sum_{j=1}^n \sum_{\ell=1}^n \hat{x}_{i} \phi^{(0)}_j \hat{v}_{ij}^{(0)} \hat{u}_{j\ell}^{(0)} \bbv_{\ell}.
	\end{align}
    Following a similar procedure, we can expand the one-shifted signal $\bbPhi^{(1)}\bbS\bbx$ and two-shifted signal $\bbPhi^{(2)}\bbS^2\bbx$ as
    \begin{align}\label{eq:GeneralizedFrequencyResponse2}
    	& \bbPhi^{(1)} \bbS \bbx = \sum_{i=1}^n \sum_{j=1}^n \sum_{\ell=1}^n \hat{x}_{i} \phi_j^{(1)} \hat{v}_{ij}^{(1)} \hat{u}_{j\ell}^{(1)} \lambda_i \bbv_{\ell},\\
    	&\bbPhi^{(2)} \bbS^2 \bbx = \sum_{i=1}^n \sum_{j=1}^n \sum_{\ell=1}^n \hat{x}_{i} \phi_j^{(2)} \hat{v}_{ij}^{(2)} \hat{u}_{j\ell}^{(2)} \lambda_i^2 \bbv_{\ell},
    \end{align}
    where $\{v_{ij}^{(1)}\}$, $\{u_{j\ell}^{(1)}\}$ and $\{v_{ij}^{(2)}\}$, $\{u_{j\ell}^{(2)}\}$ are expansion coefficients between $(\bbV$, $\bbU^{(1)})$ and $(\bbV$, $\bbU^{(2)})$. By aggregating these shifted signals, we can represent the filter output as
    \begin{align}\label{eq:GeneralizedFrequencyResponse4}
    	&\bbH_{{\rm Edge}}(\bbx, \bbS) \!=\! \sum_{i=1}^n\! \sum_{j=1}^n\! \sum_{\ell=1}^n\! \hat{x}_{i}  \!\sum_{k=0}^2 \phi_j^{(k)} \hat{v}_{ij}^{(k)} \hat{u}_{j\ell}^{(k)} \lambda_i^k \bbv_{\ell}.
    \end{align}
    By comparing \eqref{eq:GeneralizedFrequencyResponse4} with the ES-EdgeGF expression \eqref{eq:spectralESEdgeGF0}, we note that 
    the misalignment between eigenvectors $\bbV$ and $\{\bbU^{(k)}\}_{k=0}^2$ in \eqref{eq:GeneralizedFrequencyResponse4} scales the graph frequency terms by different factors $\{\hat{v}_{ij}^{(k)} \hat{u}_{j\ell}^{(k)}\}_{k=0}^2$ instead of the same one and thus, there is no common factor for extraction as in \eqref{eq:spectralESEdgeGF} of the ES-EdgeGF. By defining the scaling variables as 
    \begin{align}\label{eq:GeneralizedFrequencyResponse45}
    	\beta^{(k)} = \hat{v}_{ij}^{(k)} \hat{u}_{j\ell}^{(k)} / (\hat{v}_{ij}^{(k-1)} \hat{u}_{j\ell}^{(k-1)})
    \end{align}
    for $k=1,2$, we can re-writing \eqref{eq:GeneralizedFrequencyResponse4} as\footnote{Without loss of generality, we assume $\beta^{(k)}=\hat{v}_{ij}^{(k)} \hat{u}_{j\ell}^{(k)}$ for $k=1,2$ if $\hat{v}_{ij}^{(0)} \hat{u}_{j\ell}^{(0)}=0$ and $\prod_{a}^b = 1$ if $b < a$.}
    \begin{align}\label{eq:GeneralizedFrequencyResponse5}
    	&\bbH_{{\rm Edge}}(\bbx,\! \bbS) \!=\!\! \sum_{i=1}^n\! \sum_{j=1}^n\! \sum_{\ell=1}^n\! \hat{x}_{i} \hat{v}_{ij}^{(0)} \hat{u}_{j\ell}^{(0)}\!\! \sum_{k=0}^2\!\! \phi_j^{(k)}\!\! \prod_{\kappa=1}^k\! (\beta^{(\kappa)}\!\lambda_i\!) \bbv_{\ell}\!.
    \end{align}
    We refer to \eqref{eq:GeneralizedFrequencyResponse5} as the spectral input-output relation of \eqref{eq:exampleFilter}. It resembles \eqref{eq:spectralESEdgeGF}, but differs that the frequency $\lambda_i$ is scaled by different $\{\beta^{(k)}\}_{k=1}^2$ at different graph shifts, due to the misalignment between $\bbV$ and $\{\bbU^{(k)}\}_{k=0}^2$. The latter generates $K=2$ scaled frequencies $\{\lambda_i^{(k)} = \beta^{(k)}\lambda_i\}_{k=1}^2$ in the graph spectrum. Since the misalignment between $(\bbV$, $\bbU^{(1)})$ and $(\bbV$, $\bbU^{(2)})$ can be arbitrarily different, the scaling variables $\beta^{(1)}$, $\beta^{(2)}$ and the scaled frequencies $\lambda_i^{(1)}$, $\lambda_i^{(2)}$ can also be arbitrarily different. This indicates that a single frequency variable is not sufficient to represent the filter frequency response of \eqref{eq:exampleFilter}, and necessitates defining it as a $2$-D function
    \begin{align}\label{eq:exampleEVFrequencyResponse}
    	h_j(\bblambda) = \sum_{k=0}^2 \phi_j^{(k)} \prod_{\kappa=1}^k \lambda^{(\kappa)}~\text{for}~j=1,...,n.
    \end{align}
    The function parameters are determined by the eigenvalues of the edge weight matrices $\{\phi_j^{(k)}\}_{k=0}^2$, whereas the frequency variables $\bblambda = [\lambda^{(1)}, \lambda^{(2)}]^\top$ depend on the eigenvalues of $\bbS$ and the eigenvector misalignment between $\bbS$ and $\{\bbPhi^{(k)}\}_{k=0}^2$ -- see Fig. \ref{subfigfd} for an example. When the eigenvectors $\bbU^{(k)}$ and $\bbV$ are aligned, the scaling variables are $\beta^{(1)} \!=\! \beta^{(2)} \!=\! 1$ and the frequency variables are $\lambda_i^{(1)} \!=\! \lambda_i^{(2)}$. By imposing these constraints on \eqref{eq:exampleEVFrequencyResponse}, it reduces to a $1$-D frequency response function as that of the SI-EdgeGF [cf. \eqref{eq:SIfrequencyResponse}].
\end{defBox}

Motivated by the above example, we formally define the edge-varying filter frequency response for the general EdgeGF. 
\begin{definition}[Edge-varying filter frequency response]\label{def:generalizedFilterFrequencyResponse}
	Consider an EdgeGF with edge weight matrices $\{\bbPhi^{(k)}\}_{k=0}^K$. Let $\{\phi_i^{(k)}\}_{i=1}^n$ be the eigenvalues of $\bbPhi^{(k)}$ for $k=0,\ldots,K$. The edge-varying filter frequency response is a collection of $n$ index-specific $K$-dimensional analytic functions of the form 
	\begin{equation}\label{eq:GeneralizedFrequencyResponse}
		h_i(\bblambda) = \sum_{k=0}^K \phi_i^{(k)} \prod_{\kappa=1}^k \lambda^{(\kappa)}~\text{for}~i=1,\ldots,n,
	\end{equation}
    where $\bblambda = [\lambda^{(1)},...,\lambda^{(K)}]^\top \in \mathbb{R}^K$ is the $K$-dimensional function variable consisting of $K$ frequency variables $\{\lambda^{(k)}\}_{k=1}^K$. 
\end{definition}
\noindent The generalized filter frequency responses $\{h_i(\bblambda)\}_{i=1}^n$ are multivariate functions of a vector variable $\bblambda = [\lambda^{(1)},\ldots,\lambda^{(K)}]^\top$, where the $k$th entry $\lambda^{(k)}$ is the frequency variable resulting from the misalignment between $\bbV$ and $\bbU^{(k)}$ [cf. \eqref{eq:GeneralizedFrequencyResponse45}-\eqref{eq:GeneralizedFrequencyResponse5}]. The multi-dimensional shape of $h_i(\bblambda)$ is determined by the eigenvalues of the edge weight matrices $\{\phi_i^{(k)}\}_{k=0}^K$, while the eigenvalues of the graph shift operator $\{\lambda_i\}_{i=1}^n$ and the eigenvector misalignment between $\bbV$ and $\bbU^{(k)}$ specify the vector variable $\bblambda$. The edge-varying filter frequency responses $\{h_i(\bblambda)\}_{i=1}^n$ [cf. \eqref{eq:GeneralizedFrequencyResponse}] reduce to the eigenvector-sharing filter frequency responses $\{h_i(\lambda)\}_{i=1}^n$ [cf. \eqref{eq:ESFrequencyResponse}] when all edge weight matrices $\bbPhi^{(k)}$ share the same eigenvectors, i.e., $\bbU^{(1)} = \cdots = \bbU^{(K)} = \bbU$. In this case, the scaling factors and the frequency variables equal each other, i.e., $\beta^{(1)} = \cdots = \beta^{(K)}$ and $\lambda^{(1)} = \cdots = \lambda^{(K)}$ in \eqref{eq:GeneralizedFrequencyResponse5}. 

We now define the integral Lipschitz property for the edge-varying filter frequency response, analogous to Definitions \ref{def:LipschitzFilterFrequencyResponse}. For this purpose, we first define the Lipschitz gradient of $\{h_i(\bblambda)\}_{i=1}^n$. 
\begin{definition}[Lipschitz Gradient]\label{def:LipschitzGradient}
    Consider an edge-varying filter frequency response $h(\bblambda)$ of parameters $\{\phi^{(k)}\}_{k=0}^K$ and the variable $\bblambda = [\lambda^{(1)},\ldots,\lambda^{(K)}]^\top$ [cf. \eqref{eq:GeneralizedFrequencyResponse}]. For two specific instantiations $\bblambda_1 = [\lambda^{(1)}_1,\ldots,\lambda_1^{(K)}]^\top$ and $\bblambda_2 = [\lambda^{(1)}_2,\ldots,\lambda_2^{(K)}]^\top$ of $\bblambda$, let $\bblambda_{1:2,k} = [\lambda^{(1)}_1,\ldots,\lambda_1^{(k)}, \lambda_2^{(k+1)}, \ldots, \lambda^{(K)}_2]^\top$ be the vector that concatenates the first $k$ entries of $\bblambda_1$ and the last $K-k$ entries of $\bblambda_2$. The Lipschitz gradient of $h(\bblambda)$ over $\bblambda_1$ and $\bblambda_2$ is 
    \begin{align}
        \nabla_L h(\bblambda) |_{\bblambda_1, \bblambda_2} = \Big[ \frac{\partial h(\bblambda)}{\partial \lambda^{(1)}}\Big|_{\bblambda_{1:2, 1}}, \ldots, \frac{\partial h(\bblambda)}{\partial \lambda^{(K)}}\Big|_{\bblambda_{1:2, K}} \Big],
    \end{align}
    where $\partial h(\bblambda) / \partial \lambda^{(k)} |_{\bblambda_{1:2, k}}$ is the partial derivative of $h(\bblambda)$ w.r.t. the $k$-th entry variable $\lambda^{(k)}$ at the instantiation $\bblambda_{1:2, k}$.
\end{definition} 
\noindent The Lipschitz gradient $\nabla_L h(\bblambda) |_{\bblambda_1, \bblambda_2}$ characterizes the variability of the edge-varying filter frequency response $h(\bblambda)$ between two multivariate frequencies $\bblambda_1$ and $\bblambda_2$, i.e., $h(\bblambda_1) - h(\bblambda_2) = \nabla_L h(\bblambda_1, \bblambda_2) \cdot (\bblambda_1 - \bblambda_2)$, where $\cdot$ is the vector product. It is akin to the derivative of the shift-invariant / eigenvector-sharing filter frequency response $h'(\lambda)$ [cf. \eqref{eq:SIfrequencyResponse}], and allows to define the integral Lipschitz EdgeGF as follows.
\begin{definition}[Integral Lipschitz EdgeGF]\label{def:LipschitzEVGF}
	Consider an EdgeGF with the edge-varying filter frequency response $\{h_i(\bblambda)\}_{i=1}^n$ [cf. \eqref{eq:GeneralizedFrequencyResponse}]. The EdgeGF is integral Lipschitz if there exists a constant $C_L > 0$ such that for any $\bblambda_1$ and $\bblambda_2$, it holds that 
	\begin{align}\label{eq:LipschitzGeneralizedResponse}
		\Big|\frac{\bblambda_1 + \bblambda_2}{2} \cdot \nabla_L h_i(\bblambda) |_{\bblambda_1, \bblambda_2}\Big| \le C_L~\text{for}~i=1,\ldots,n
	\end{align}
    with $\nabla_L h_i(\bblambda) |_{\bblambda_1, \bblambda_2}$ the Lipschitz gradient over $\bblambda_1$ and $\bblambda_2$. 
\end{definition} 
\noindent Integral Lipschitz EdgeGFs restrict the variability of the edge-varying filter frequency response. Analogous to integral Lipschitz SI-EdgeGFs / ES-EdgeGFs [Def. \ref{def:LipschitzFilterFrequencyResponse}], the edge-varying filter frequency response may change fast when $\bblambda$ is instantiated at small values, but is required to change more slowly when $\bblambda$ is instantiated at large values. 

\begin{remark}\label{remark:EdgeGFFrequency}
	For a fixed or given graph $\bbS$, the EdgeGF instantiates each $h_j(\bblambda)$ on $n$ multivariate frequencies $\{\bblambda_i\}_{i=1}^n$ because of the eigenvector misalignment between $\{\bbPhi^{(k)}\}_{k=0}^K$ and $\bbS$ [cf. \eqref{eq:GeneralizedFrequencyResponse5}]. Thus, integral Lipschitz EdgeGFs require each $h_j(\bblambda)$ satisfying the condition \eqref{eq:LipschitzGeneralizedResponse} w.r.t. all instances $\{\bblambda_i\}_{i=1}^n$ rather than a single one $\bblambda_j$.
\end{remark}

\subsection{Stability of EdgeGFs}

We now show the stability of the general EdgeGFs. 
\begin{theorem}[Stability of EdgeGF]\label{theorem:EVfilterStability}
	Consider the EdgeGF $\bbH_{\rm Edge}(\bbx, \bbS)$ with the edge weight matrices $\{\bbPhi^{(k)}\}_{k=0}^K$ [cf. \eqref{eq:EdgeNet0}]. Let $h_i(\bblambda)$ be the edge-varying frequency response that is integral Lipschitz w.r.t. $C_L$ and satisfies $|h_i(\bblambda)| \le 1$ for $i=1,\ldots,n$ [cf. \eqref{eq:LipschitzGeneralizedResponse}], $\bbU^{(k)}$ and $\bbV$ be the eigenvectors of $\bbS$ and $\bbPhi^{(k)}$ that are $\varepsilon$-misaligned for $k=0,\ldots,K$ [Def. \ref{def:misalignment}], and $\widetilde{\bbS}$ be the perturbed graph with perturbation size $\eps$ [cf. \eqref{def:relativePerturbation}]. Then, for any graph signal $\bbx$, it holds that 
	\begin{align}\label{eq:EVFilterStability}
		&\| \bbH_{\rm Edge}(\bbx, \bbS) - \bbH_{\rm Edge}(\bbx, \widetilde{\bbS}) \|_2 \\
		&\le C_{\rm Edge} \| \bbx \|_2 \eps + \ccalO\big(\eps^2\big) + \ccalO(\varepsilon^2), \nonumber
	\end{align}
	where $C_{\rm Edge} = 2\sqrt{n}(1 + 2n\varepsilon) C_L$ is the stability constant.
\end{theorem}

\begin{proof}
    See Appendix \ref{proof:theorem3}.
\end{proof}

Theorem \ref{theorem:EVfilterStability} shows that the stability constant increases from $2\sqrt{n}(1 + n\varepsilon) C_L$ of the ES-EdgeGFs [cf. \eqref{eq:EI-EVFilterStability}] to $2\sqrt{n}(1 + 2n\varepsilon) C_L$. This is because the edge weight matrices $\bbPhi^{(k)}$ no longer share the same eigenvectors but have different ones $\bbU^{(k)}$. The latter increases the eigenvector misalignment among the edge weight matrices $\bbPhi^{(k)}$ and the underlying graph $\bbS$, which, in turn, amplifies the impact of frequency deviations induced by graph perturbations when propagating through the filter. Moreover, the condition required for the stability analysis is stronger for the general EdgeGFs. Specifically, we now require the Lipschitz continuity of edge-varying filter frequency responses $\{h_i(\bblambda)\}_{i=1}^n$ over the multivariate frequency $\bblambda$, while the ES-EdgeGFs need only the Lipschitz continuity over the univariate frequency $\lambda$. Both aspects emphasize the trade-off between the DoFs of edge weight matrices and the stability of graph neural networks, i.e., arbitrary edge weight matrices without eigenvector restrictions increase the parameter space and improve the representational capacity, but degrade the architecture stability.


{\linespread{1}
	\begin{table*}[]
		
		\scriptsize
		
		\centering
		
		\renewcommand\tabcolsep{3pt}
		
		\caption{Stability of SI-EdgeNets, ES-EdgeNets and EdgeNets to graph perturbations. The constraint on the edge weight matrices reduces and the degrees of freedom increase from SI-EdgeNets, ES-EdgeNets to EdgeNets, indicating an increasing of representational capacity. The stability bound increases likewise, implying a degradation of architecture robustness. It shows an explicit trade-off between stability and representation capacity in GNNs.   
		}
		
		\label{Tab:stabilitySummary}
		
		\begin{tabular}{cccc}
			
			\toprule
			
			\multirow{1}{*}{Architecture} & \multicolumn{1}{c}{Edge weight matrices} & \multicolumn{1}{c}{Stability bound} & \multicolumn{1}{c}{Examples} 
			\\
			
			\midrule
			
			SI-EdgeNet [cf. \eqref{eq:SIEVGF}]    & $\bbPhi \!\in \! \big\{\bbV \bbLambda \bbV^{-1}\!, [{\rm vec}(\bbPhi)]_i \!=\! 0~\text{for}~i \!\in\! \ccalA \big\}$         & $2\sqrt{n}C_L \| \bbx \|_2 \eps \!+\! \ccalO\big(\eps^2\big)$        & GCNN \cite{gama2018convolutional}, ChebNet \cite{defferrard2016convolutional}, JKNet \cite{xu2018representation}, GIN \cite{xupowerful}, etc.  
			\\[5pt]
			
			ES-EdgeNet [cf. \eqref{eq:EIEVGF}]    & $\bbPhi \!\in\! \cup_\bbU \big\{\bbU \bbLambda \bbU^{-1}\!, [{\rm vec}(\bbPhi)]_i \!=\! 0~\text{for}~i \!\in\! \ccalA \big\}$                    & $2\sqrt{n}(1+n\varepsilon) C_L \| \bbx \|_2 \eps \!+\! \ccalO\big(\eps^2\big)$   & NV-GNN \cite{gama2022node}, GAT \cite{velivckovicgraph}, NGN \cite{de2020natural}, etc.      
			\\[5pt]
			
			EdgeNet [cf. \eqref{eq:EdgeNet0}]       & $\bbPhi \!\in \! \big\{[{\rm vec}(\bbPhi)]_i \!=\! 0~\text{for}~i \!\in\! \ccalA \big\}$                    &                   $2\sqrt{n}(1 \!\!+\! 2n\varepsilon\!) C_L \| \bbx \|_2 \eps \!+\! \ccalO\big(\eps^2\!\big) \!+\! \ccalO(\varepsilon^2\!)$     & GCAT \cite{isufi2021edgenets}, Hybrid GNN \cite{isufi2021edgenets}, etc. 
			\\[5pt]
			
			\bottomrule
			
		\end{tabular}
		
\end{table*}}


We see from Theorems \ref{theorem:SIEVfilterStability}-\ref{theorem:EVfilterStability} that not only the magnitude of eigenvector misalignment, i.e., the value of $\varepsilon$, but also the number of misaligned eigenvectors $\{\bbU^{(k)}\}_{k=0}^K$ degrade the stability of the architecture. Put differently, despite the strongest representational capacity, the general EdgeGFs do not necessarily perform better than the SI-EdgeGFs / ES-EdgeGFs under graph perturbations -- see Section \ref{sec:experiments} for experimental corroborations. 

\begin{remark}%
    When the eigenvectors of the edge weight matrices $\{\bbPhi^{(k)}\}_{k=0}^K$ align with those of the underlying graph $\bbS$, i.e., $\bbU^{(k)} = \bbV$ for $k=0,\ldots,K$, the expansion coefficients between $\{\bbU^{(k)}\}_{k=0}^K$ and $\bbV$ equal one. The EdgeGF reduces to the SI-EdgeGF, the edge-varying filter frequency response $h_i(\bblambda)$ in \eqref{eq:GeneralizedFrequencyResponse} reduces to the filter frequency response $h_i(\lambda)$ in \eqref{eq:SIfrequencyResponse}, and the stability result of the EdgeGF [cf. \eqref{eq:EVFilterStability}] recovers that of the SI-EdgeGF [cf. \eqref{eq:SIEVFilterStability}].
\end{remark}

\subsection{Stability of EdgeNets}\label{subsec:stabilityEdgeNet}

The stability of SI-EdgeNets, ES-EdgeNets and EdgeNets inherits from that of their respective filters, which needs an additional assumption on the nonlinearity. 
\begin{assumption}[Lipschitz nonlinearity]\label{asp:nonlinearity}
    The nonlinearity $\sigma(\cdot)$ satisfying $\sigma(0)=0$ is normalized Lipschitz, i.e., for any $a,b \in \mathbb{R}$, it holds that $|\sigma(a) - \sigma(b)| \le |a-b|$.
\end{assumption}
\noindent Assumption \ref{asp:nonlinearity} is commonly satisfied and examples include ReLU and absolute value. The following theorem formally 
analyzes the stability of EdgeNets. 
\begin{theorem}[Stability of EdgeNet]\label{thm:stabilityEdgeNet}
	Consider an EdgeNet $\bbPsi_{\rm Edge}(\bbx,\bbS,\ccalH)$ of $L$ layers and $F$ features. 
	[cf. \eqref{eq:EdgeNet2}-\eqref{eq:EdgeNet3}]. Assume the same setting as in Theorem \ref{theorem:EVfilterStability} and the nonlinearity satisfies Assumption \ref{asp:nonlinearity}. Let $\widetilde{\bbS}$ be the perturbed graph with perturbation size $\eps$ [cf. \eqref{def:relativePerturbation}]. Then, for any graph signal $\bbx$, it holds that
	\begin{align}\label{eq:stabilityEdgeNet}
		&\Big\|\bbPsi_{\rm Edge}(\bbx,\bbS,\ccalH) - \bbPsi_{\rm Edge}(\bbx,\widetilde{\bbS},\ccalH)\Big\|_2 \\
		&\le L F^{L-1} C_{\rm Edge} \|\bbx\|_2 \eps + \ccalO(\varepsilon^2) + \ccalO(\eps^2). \nonumber
	\end{align}
	The stability of SI-EdgeNets or ES-EdgeNets can be obtained by replacing $C_{\rm Edge}$ [cf. \eqref{eq:EVFilterStability}] with $C_{\rm SI}$ [cf. \eqref{eq:SIEVFilterStability}] or $C_{\rm ES}$ [cf. \eqref{eq:EI-EVFilterStability}]. 
\end{theorem}

\begin{proof}
    See Appendix \ref{proof:theorem4}.
\end{proof}
Theorem \ref{thm:stabilityEdgeNet} states the effects of EdgeNet hyperparameters on the architecture stability. Specifically, the stability bound increases with the number of layers $L$ and features $F$. This indicates an extra trade-off between representational capacity and perturbation stability, i.e., a deeper and wider EdgeNet may increase the expressive power but becomes also less stable.\footnote{We consider the filter frequency response is normalized bounded $|h(\lambda)|\le 1$ (or $|h(\bblambda)|\le 1$) and the nonlinearity is normalized Lipschitz [Def. \ref{asp:nonlinearity}] to keep the bounds light on notation. If these constants are not one, they can be easily included in the stability constants following the proofs.} This result holds for any GNNs in the framework of EdgeNets, e.g., GCNNs, NV-GNNs and GATs, where the only difference lies in the stability constant inherited from their respective filters.

\begin{figure*}%
	\centering
	\begin{subfigure}{0.63\columnwidth}
		\includegraphics[width=1.05\linewidth, height = 0.76\linewidth, trim = {0cm 0cm 0cm 1cm}, clip]{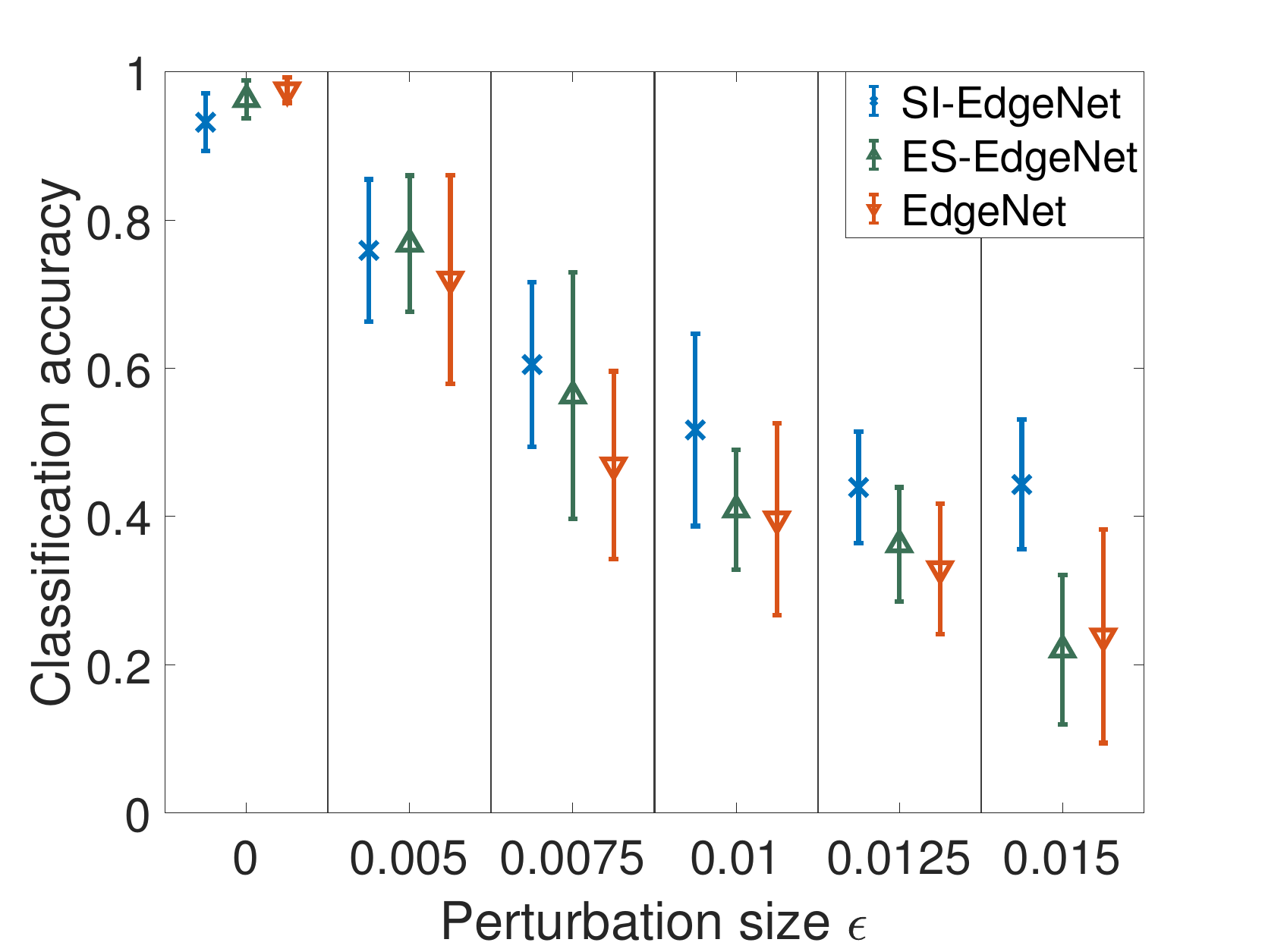}%
		\caption{}%
		\label{subfiga}%
	\end{subfigure}\hfill\hfill%
	\begin{subfigure}{0.63\columnwidth}
		\includegraphics[width=1.05\linewidth,height = 0.76\linewidth, trim = {0cm 0cm 0cm 1cm}, clip]{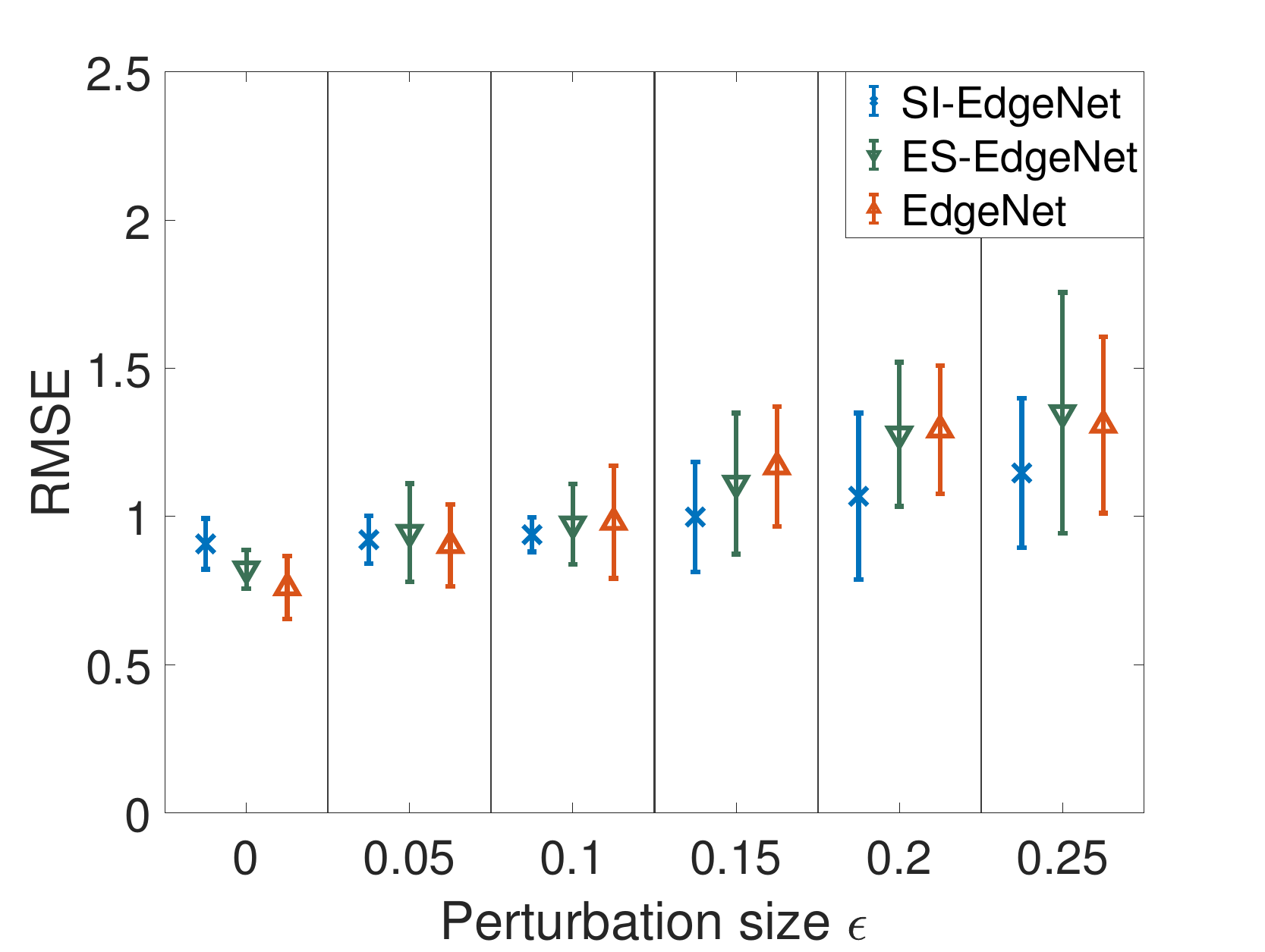}%
		\caption{}%
		\label{subfigb}%
	\end{subfigure}\hfill\hfill%
	\begin{subfigure}{0.63\columnwidth}
		\includegraphics[width=1.05\linewidth,height = 0.76\linewidth, trim = {0cm 0cm 0cm 1cm}, clip]{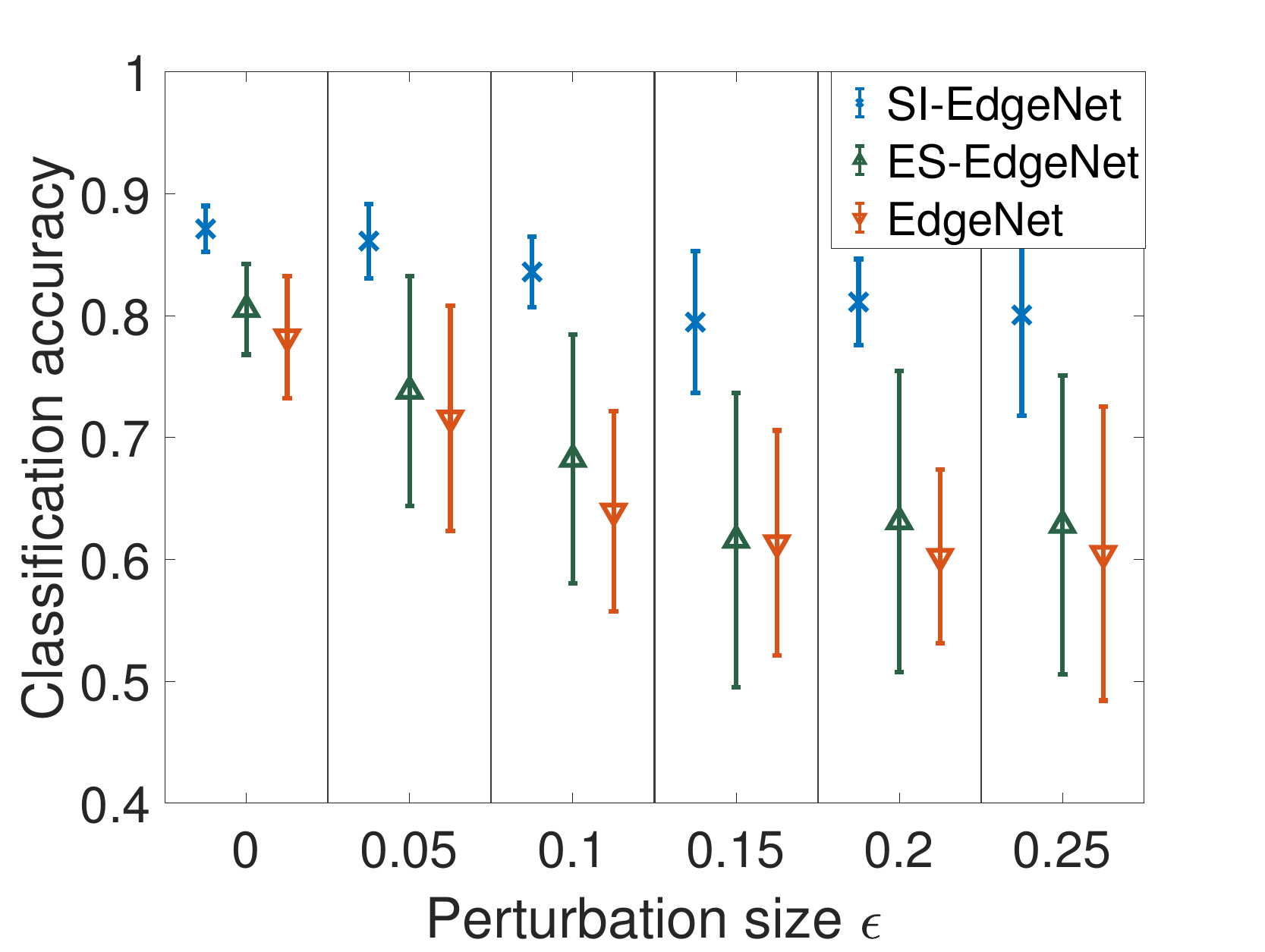}%
		\caption{}%
		\label{subfigc}%
	\end{subfigure}%
	\caption{Performance of SI-EdgeNet, ES-EdgeNet and EdgeNet under perturbations with 
		different sizes in (a) source localization, (b) movie recommendation and (c) authorship attribution.}\label{fig:performance}
\end{figure*} 

\smallskip
\noindent \textbf{Discussion.} SI-EdgeNets require the edge weight matrices $\{\bbPhi^{(k)}\}_{k=0}^K$ to share the same eigenvectors as the underlying graph $\bbS$. ES-EdgeNets relax the constraint by requiring the same eigenvectors among $\{\bbPhi^{(k)}\}_{k=0}^K$ but different from $\bbS$, while EdgeNets allows $\{\bbPhi^{(k)}\}_{k=0}^K$ with arbitrarily different eigenvectors. The DoFs in the architecture parameter space increase from SI-EdgeNets, ES-EdgeNets to EdgeNets, which improve the representational capacity. However, the relaxed constraint results in an eigenvector misalignment between $\{\bbPhi^{(k)}\}_{k=0}^K$ and $\bbS$, and the degree of eigenvector misalignment increases from ES-EdgeNets to the general EdgeNets. This amplifies the effect of spectral deviations caused by topological perturbations, when propagating through the filter, and degrades the architecture stability from two aspects: (i) the stability bound increases from \eqref{eq:SIEVFilterStability}, \eqref{eq:EI-EVFilterStability} to \eqref{eq:EVFilterStability}, and (ii) the sufficient conditions required for the stability analysis become stronger from the Lipschitz property of univariate frequency responses \eqref{eq:SIfrequencyResponse} to multivariate ones \eqref{eq:GeneralizedFrequencyResponse}. These theoretical findings show an explicit trade-off between the representational capacity and the stability against perturbations, and indicate that it is not always good to increase the parameter space, i.e., simpler GNNs may achieve a better trade-off in certain applications. Table \ref{Tab:stabilitySummary} summarizes our results.

\section{Experiments}\label{sec:experiments}

We corroborate our theoretical findings on synthetic and real-world classification and regression problems. In all experiments, we train with the ADAM optimizer of decaying factors $\beta_1\!=\!0.9$ and $\beta_2 \!=\! 0.999$ \cite{kingma2014adam}. 

\subsection{Experimental Setup}

We consider three experiments: source localization, movie recommendation, and authorship attribution. 

\smallskip
\noindent \textbf{Source localization.} The goal is to find the source community where a diffused signal is originated over a stochastic block model (SBM) graph. The graph consists of $100$ nodes equally divided into $10$ communities $\{c_1,\ldots,c_{10}\}$, where the inter- and intra-community edge probability is $0.2$ and $0.8$. The source node $s_i$ of the community $c_i$ is the node with the largest degree for $i=1,\ldots,10$. The source signal is a Kronecker delta originating at one of the sources and diffused at time $t$ as $\bbx_t = \bbS^t \bbx + \bbn_t$ where $\bbS$ is the normalized adjacency matrix and $\bbn_t$ is the zero-mean Gaussian noise. The dataset contains $5000$ samples for training, $250$ for validation and $250$ for testing, where each sample takes the form of $(\bbx_t, c_i)$ with random $t \in \mathbb{Z}_+$ and $i \in \{1,\ldots,10\}$. The classification accuracy is averaged over $10$ graph realizations with $10$ random data splits. 

\begin{figure*}%
	\centering
	\begin{subfigure}{0.63\columnwidth}
		\includegraphics[width=1.05\linewidth, height = 0.76\linewidth, trim = {0cm 0cm 0cm 1cm}, clip]{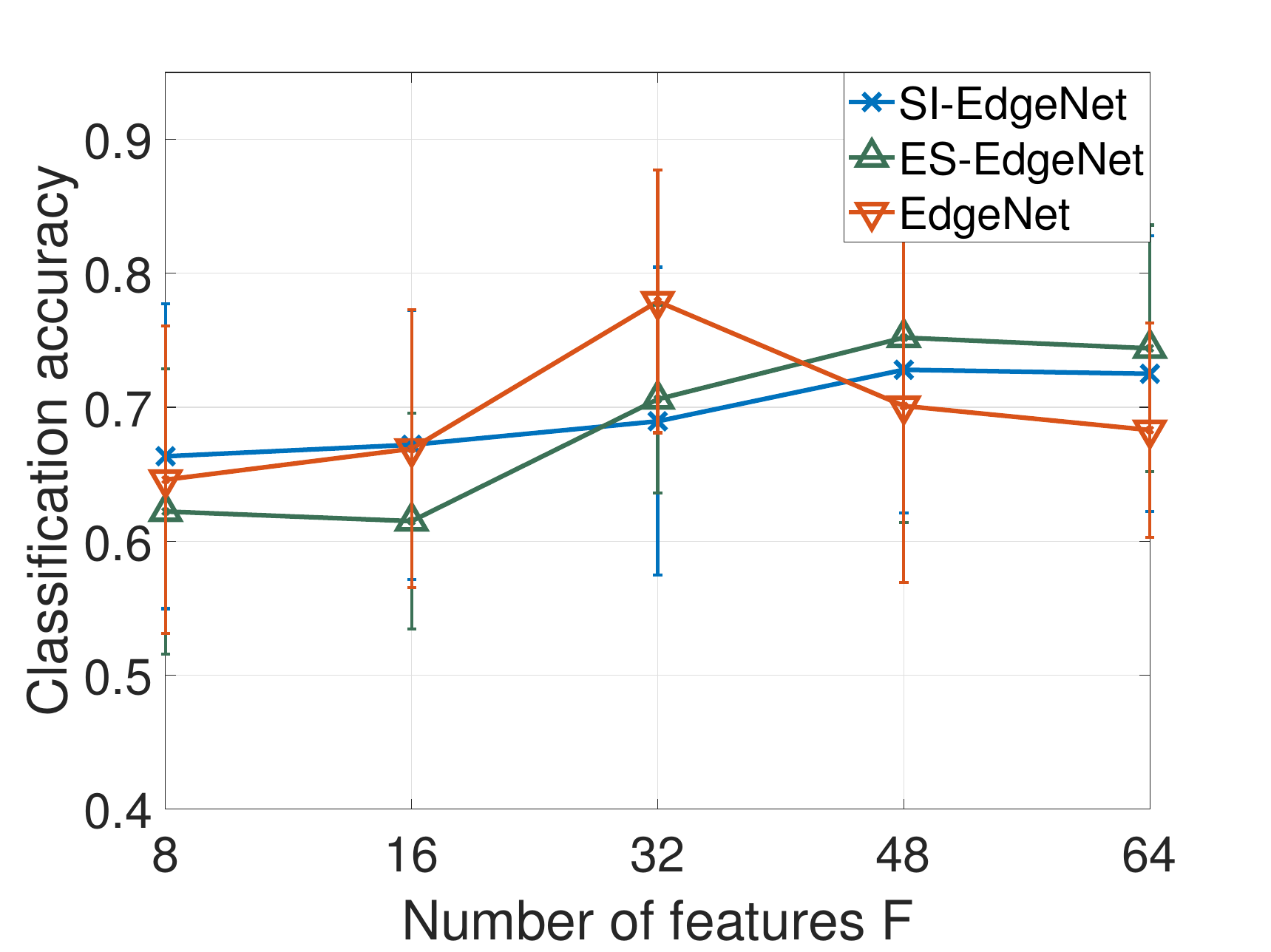}%
		\caption{}%
		\label{subfigF}%
	\end{subfigure}\hfill\hfill%
	\begin{subfigure}{0.63\columnwidth}
		\includegraphics[width=1.05\linewidth,height = 0.76\linewidth, trim = {0cm 0cm 0cm 1cm}, clip]{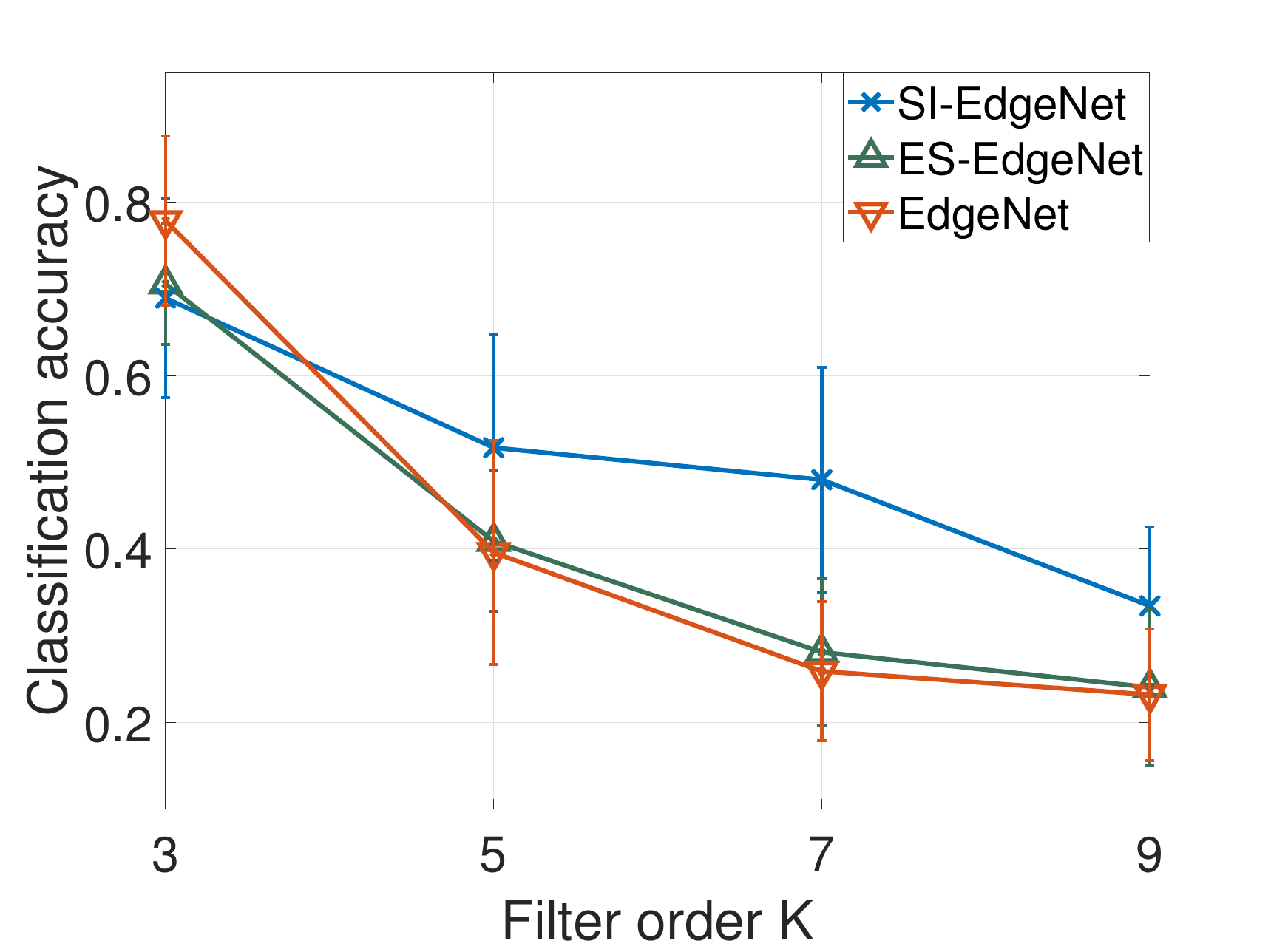}%
		\caption{}%
		\label{subfigK}%
	\end{subfigure}\hfill\hfill%
	\begin{subfigure}{0.63\columnwidth}
		\includegraphics[width=1.05\linewidth,height = 0.76\linewidth, trim = {0cm 0cm 0cm 1cm}, clip]{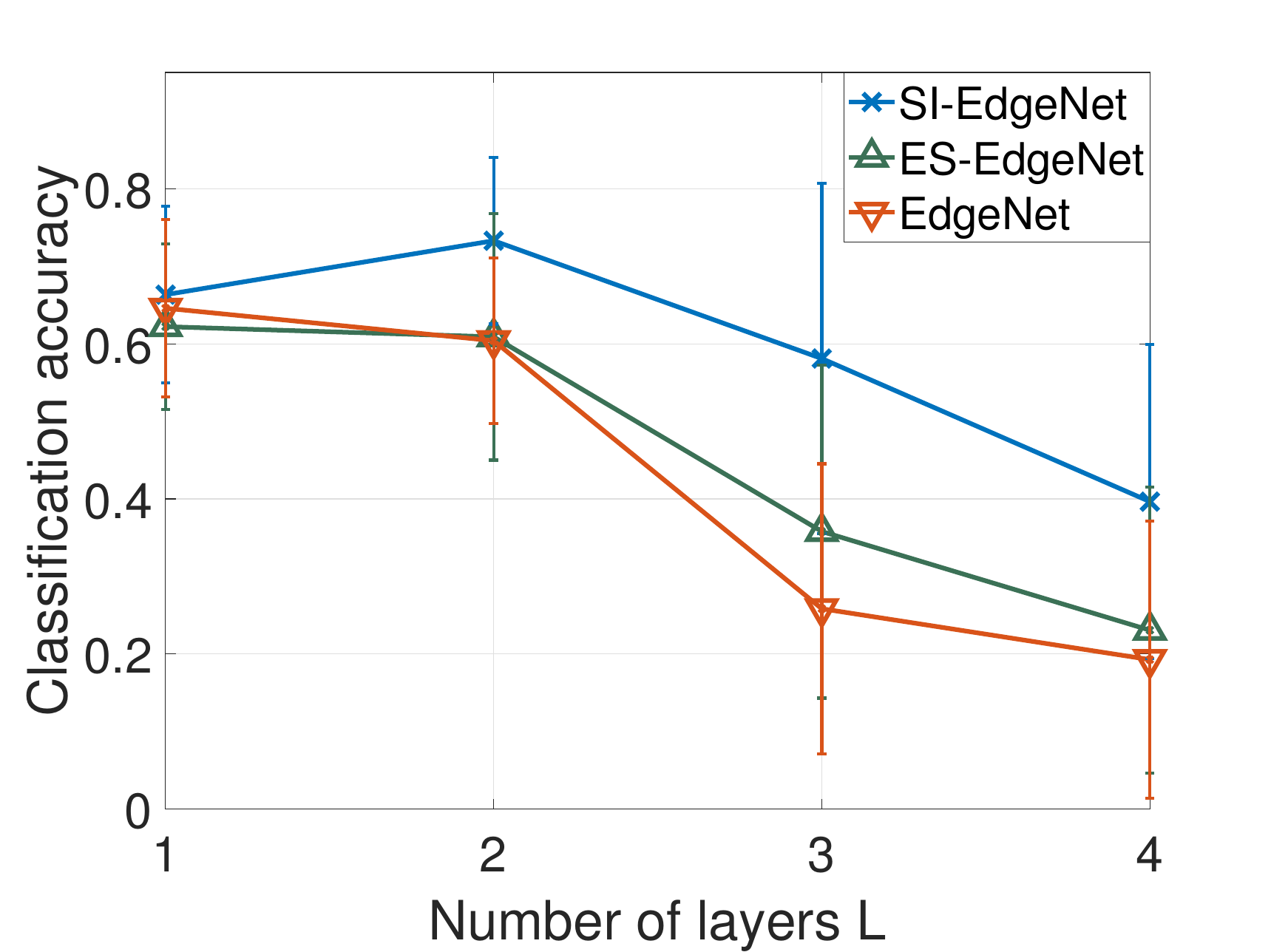}%
		\caption{}%
		\label{subfigL}%
	\end{subfigure}%
	\caption{Performance of SI-EdgeNet, ES-EdgeNet and EdgeNet with different hyperparameters in source localization. (a) Number of features $F$. (b) Filter order $K$. (c) Number of layers $L$.}\label{fig:performanceHyperparameter}
\end{figure*} 

\smallskip
\noindent \textbf{Movie recommendation.} The goal is to predict movie ratings in the MovieLens100K dataset, which consists of $943$ users, $1682$ movies, and $100$K out of the $1.5$M potential ratings \cite{harper2015movielens}. We consider the movie similarity graph with nodes as movies and edges as Pearson similarities, and prune the graph to keep the top-ten most similar connections per node. The graph signal is the movie ratings given by a user, yielding $943$ signals. We are interested in predicting the ratings of the movie \emph{Star Wars}, which is the one with the largest number of ratings given by users. It is equivalent to completing the corresponding column of the user-item rating matrix. We split the dataset as $90\%$ for training and $10\%$ for testing. The performance is measured with the RMSE and averaged over $10$ random splits.

\smallskip
\noindent \textbf{Authorship attribution.} The goal is to classify if a text excerpt belongs to a specific author or any of the other $20$ contemporary authors in the authorship dataset \cite{segarra2015authorship}. We build the graph based on Word Adjacency Networks (WANs), where each node is a function word without semantic meaning (e.g., prepositions, pronouns, etc.) and each edge weight represents the average co-occurence of function words, discounted by the relative distance (i.e. if two words are next to each other, the weight is higher than that if two words are further apart). The graph signal is the frequency count of function words in the text excerpt of $1000$ words. We are interested in classifying the author \emph{Edgar Allan Poe}, and the dataset consists of $740$ samples for training, $64$ for validation and $42$ for testing. The performance is measured with the classification accuracy and averaged over $10$ random splits.

As default, we consider an architecture comprising a cascade of a graph filtering layer with ReLU nonlinearity and a fully connected layer with softmax nonlinearity, where the number of EdgeGFs is $F=32$ and the filter order is $K=5$. Different architectures with varying hyperparameters are evaluated in Section \ref{subsec:hyperparameter}.

\subsection{Stability-Performance Trade-Off}

We evaluate the performance of SI-EdgeNet, ES-EdgeNet and EdgeNet in three experiments with different perturbation sizes $\eps$, and show the results in Fig. \ref{fig:performance}. For the source localization and movie recommendation in Figs. \ref{subfiga}-\ref{subfigb}, EdgeNet outperforms ES-EdgeNet and SI-EdgeNet when there is no graph perturbation, i.e., perturbation size $\eps=0$. We attribute this behavior to its strongest representational capacity. All architectures experience performance degradation under perturbations as expected, which is emphasized with the increase of the perturbation size. SI-EdgeNet performs best under large perturbations because of its strongest stability, while EdgeNet exhibits the worst performance and is most vulnerable to perturbations. ES-EdgeNet performs better than SI-EdgeNet without perturbations and suffers from less performance degradation than EdgeNet under perturbations. This follows our results in Theorems \ref{theorem:SIEVfilterStability}-\ref{theorem:EVfilterStability}, i.e., the stability improves with the alignment between the eigenvectors of edge weight matrices and the underlying graph, highlighting the trade-off between stability and representational capacity. 

For the authorship attribution in Fig. \ref{subfigc}, SI-EdgeNet achieves the best performance without graph perturbation. This is because the representational capacity of SI-EdgeNet is sufficient to solve the binary classification problem of authorship attribution and the increased architecture parameters of ES-EdgeNet or EdgeNet may require more data for training to obtain a better performance; hence, showing their incurred parameters may be redundant for the task. SI-EdgeNet exhibits also the strongest stability under graph perturbations, which aligns with the source localization and movie recommendation results in Figs. \ref{subfiga}-\ref{subfigb}. The result indicates that redundant representational capacity may be unnecessary or even harmful, and that (relatively) simple GNNs can achieve both a good performance and a strong stability in certain applications.

\subsection{Hyperparameter Sensitivity}\label{subsec:hyperparameter}

\begin{figure*}%
	\centering
	\begin{subfigure}{0.63\columnwidth}
		\includegraphics[width=1.05\linewidth, height = 0.76\linewidth, trim = {0cm 0cm 0cm 1cm}, clip]{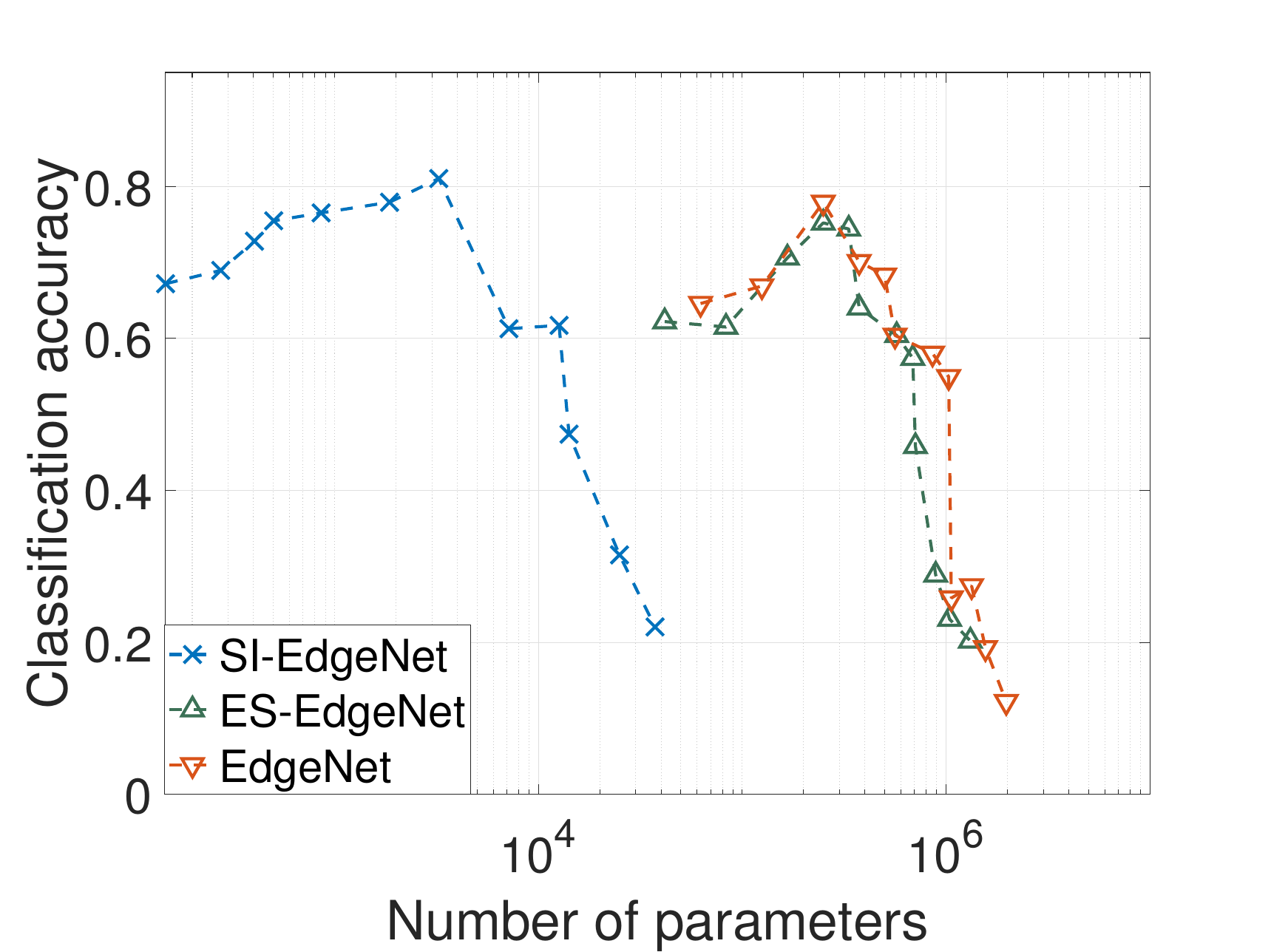}%
		\caption{}%
		\label{subfig4a}%
	\end{subfigure}\hfill\hfill%
	\begin{subfigure}{0.63\columnwidth}
		\includegraphics[width=1.05\linewidth,height = 0.76\linewidth, trim = {0cm 0cm 0cm 1cm}, clip]{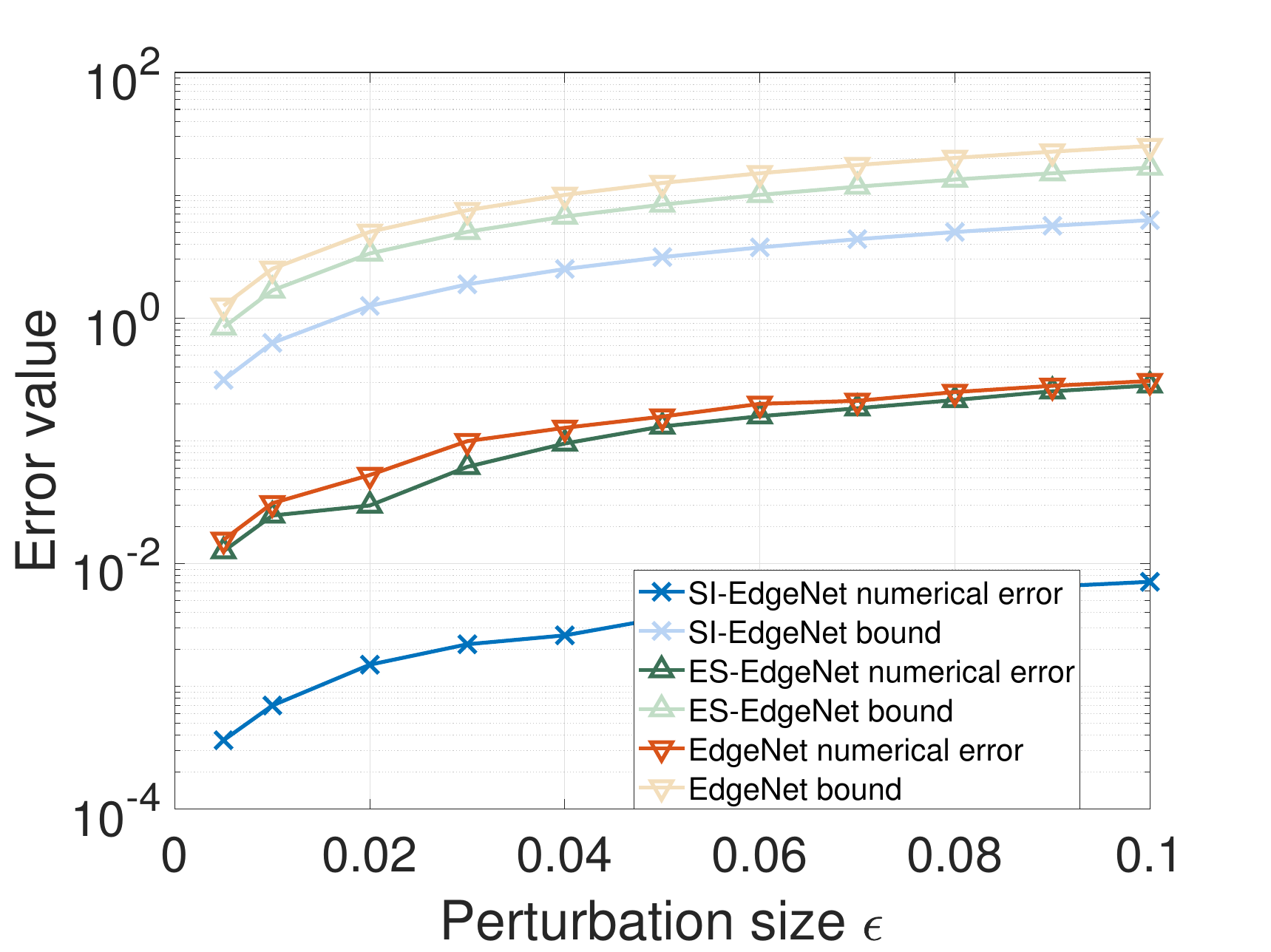}%
		\caption{}%
		\label{subfig4b}%
	\end{subfigure}\hfill\hfill%
	\begin{subfigure}{0.63\columnwidth}
		\includegraphics[width=1.05\linewidth,height = 0.76\linewidth, trim = {0cm 0cm 0cm 1cm}, clip]{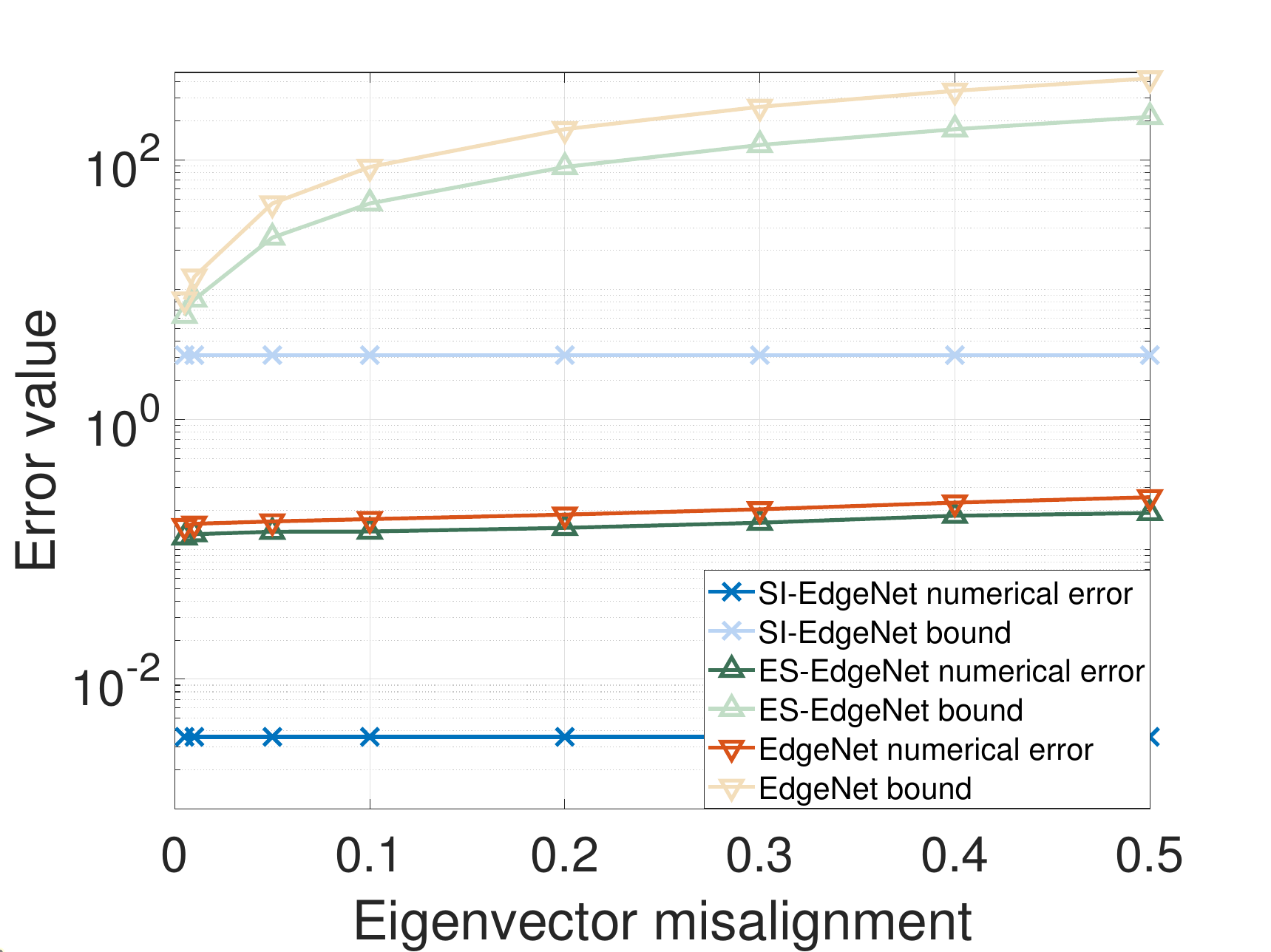}%
		\caption{}%
		\label{subfig4c}%
	\end{subfigure}%
	\caption{(a) Performance of SI-EdgeNet, ES-EdgeNet and EdgeNet with different numbers of parameters. (b) Comparison between numerical error and theoretical bound of SI-EdgeNet, ES-EdgeNet and EdgeNet with different perturbation sizes $\epsilon$. (c) Comparison between numerical error and theoretical bound of SI-EdgeNet, ES-EdgeNet and EdgeNet with different eigenvector misalignment $\varepsilon$.}\label{fig:bound}
\end{figure*} 

Then, we 
study the effect of different hyperparameters on the stability-performance trade-off. To change the representation capacity of the architectures, we vary the number of features $F$ in Fig. \ref{subfigF}, the filter order $K$ in Fig. \ref{subfigK}, and the number of layers $L$ in Fig. \ref{subfigL}. The perturbation size is $0.01$. Fig. \ref{subfigF} shows that the performance of all architectures first increases with the number of features $F$, and then decreases for large $F$. This is because an initial increase in $F$ improves the representational capacity that contributes to a better performance, while increasing it further starts inducing more perturbations than improving the representational capacity and it leads to a less stable architecture as indicated in Theorem \ref{thm:stabilityEdgeNet}. This demonstrates how the stability trades with the representational capacity, i.e., the capacity improvement first dominates the stability degradation and then the domination opposites as $F$ increases. Moreover, the performance of EdgeNet decreases earlier than SI-EdgeNet and ES-EdgeNet since the EdgeNet reaches the sufficient representational capacity with a smaller $F$, and the larger but redundant representational capacity does not improve the performance but degrades the stability. 

Figs. \ref{subfigK}-\ref{subfigL} show that the performance decreases with the filter order $K$ and the number of layers $L$. This can be explained by the fact that a small $K$ or $L$ has achieved a sufficient representational capacity. Further increasing it involves more perturbation effects and results in a larger output difference. This corroborates the findings in Theorem \ref{thm:stabilityEdgeNet}, where more layers include more filters that amplify the perturbation impact and larger filter orders involve higher order terms that increase the Lipschitz constant. From the results in Fig. \ref{fig:performanceHyperparameter}, we see that SI-EdgeNet tends to achieve the best performance with more features and layers, e.g., $F=48$ and $L=2$, while ES-EdgeNet and EdgeNet tend to perform best with less ones, e.g., $F=32$ and $L=1$. This is in line with our expectation that SI-EdgeNet requires more features / layers for sufficient expressive power, while ES-EdgeNet and EdgeNet can achieve sufficient expressive power with less features / layers and are more stable in that circumstance.

Fig. \ref{subfig4a} shows the classification accuracy as a function of the number of trainable parameters. For all architectures, as the number of parameters increases, the performance first increases and then decreases. This is because the representational capacity may not be sufficient for small architectures and the capacity improvement first dominates the stability degradation. However, when the representational capacity reaches a saturated level, the opposite happens and the performance then decreases. These results further corroborate our theoretical findings regarding the trade-off between stability and representational capacity. 

\subsection{Stability Bound Corroboration}\label{subsec:author}

Lastly, we corroborate the stability analysis in Theorems \ref{theorem:SIEVfilterStability}-\ref{thm:stabilityEdgeNet} by comparing the empirical output difference of the EdgeNet and the theoretical bound in \eqref{eq:stabilityEdgeNet} with different perturbation sizes $\eps$ in Fig. \ref{subfig4b} and different eigenvector misalignment $\varepsilon$ in Fig. \ref{subfig4c}. We consider the SBM graph and the EdgeNet of $L=2$ layers, each layer containing $F=2$ filters of order $K=3$ followed by the ReLU nonlinearity. The graph signal is with unitary energy. 

We see that the theoretical analysis yields fair bounds for the numerical output difference, and the results evidence the increase of the output difference, i.e., the degradation of the stability, from SI-EdgeNet, ES-EdgeNet to EdgeNet as shown in Theorems \ref{theorem:SIEVfilterStability}-\ref{theorem:EVfilterStability}. Moreover, it corroborates the impact of the perturbation size $\eps$ and the eigenvector misalignment $\varepsilon$, i.e., the output difference of all architectures increases with $\eps$ and $\varepsilon$. The increasing rate of the output difference with $\eps$ is larger than that with $\varepsilon$, which implies that the perturbation size has a greater effect than the eigenvector misalignment. We also note that the bounds are not tight, essentially because these bounds hold uniformly for all graphs and signals. However, they represent quite well the degradation rate. Tighter bounds could be found if one tailors the stability analysis to a specific graph rather than providing a universal one as we do here. 

\section{Discussion}\label{sec:conclusion}

This paper studied the stability of EdgeNets, a general framework that unifies state-of-the-art GNNs including in particular graph convolutional and attentional networks, among others. We proved that the general EdgeNets are stable to relative perturbations in the support, which confirmed that stability to graph perturbations is a property that holds for a broad class of existing GNN architectures and even for EdgeNet subcategories not explicitly studied before, rather than being limited to some particular cases.

By comparing the conditions required to guarantee stability and the resulting stability bounds, we showed that architectures with more flexibility in their parameter space are less stable. Specifically, we taxonomized the subclasses of EdgeNets with increasing representational capacity as:
\begin{enumerate}
	\item SI-EdgeNets: shift-invariant solutions whose all edge varying parameter matrices share the eigenvectors with the graph shift operator. Graph convolutional neural network is a particular case in this category.
	\item ES-EdgeNets: eigenvector-sharing solutions whose all edge varying parameter matrices share the eigenvectors among themselves but the latter can be different from those of the graph shift operator. Graph attention networks and natural graph networks are particular cases in this category.
	\item EdgeNets: general solutions whose edge varying parameter matrices can have arbitrary eigenvectors.
\end{enumerate}
This taxonomy links to the degrees of freedom of the trainable edge weight matrices and thus, affects their representational capacity. Hinging on the spectral behavior of these categories, we showed that stability could be guaranteed for the SI-EdgeNets with fewer restrictions on the variability of the filter frequency response, and with increasing restrictions on the ES-EdgeNets and the general EdgeNets. The resulting stability bound also increases from SI-EdgeNets, ES-EdgeNets to the general EdgeNets. Consequently, we established a stability-representational capacity trade-off that we corroborated numerically. Our theoretical results and numerical experiments showed that simpler solutions should be preferred over more expressive ones if the graph is uncertain (perturbed) or when the task is simple.

Our study has limitations. First, it is conclusive only for small relative perturbations on the topology but not for more aggressive ones such as link or vertex removal. While extending to random link changes could be feasible following \cite{gao2021stochastic}, the stability of GNNs to vertex removals is challenging to study. Second, we used the degrees of freedom of edge weight matrices as a proxy of the representational capacity to study different GNNs. While this rationale is valid, it is possible to improve the representational capacity with other mechanisms. Nevertheless, the findings of this paper are relevant to understanding the advantages and limitations of GNNs, especially in light of their stability.

Future work could be focused on overcoming the above limitations to study the stability under aggressive vertex removal, where the key but challenging step is to provide the conditions that are needed for stability analysis. Another future work arising from this study is to design GNN solutions that strike a balance between stability and representational capacity. For example, one may penalize undesired properties (e.g., the eigenvector misalignment) when training the general EdgeNets, rather than imposing hard constraints on the edge weight matrices (e.g., shift invariance of GCNNs). All in all, our findings shed light on the inner-working mechanisms of the GNN architectures and could lead to different perspectives when designing new solutions.

\appendices 

\section{Shift Invariant Edge Weight Matrices}\label{appendix:shiftInvariance}

The SI-EdgeGF is invariant to graph shift operations [cf. (5)], i.e., 
\begin{align}\label{eq:SI}
	\bbH_{\rm Edge}(\bbx, \bbS)\bbS = \bbS \bbH_{\rm Edge}(\bbx, \bbS)
\end{align}
for any graph signal $\bbx$. While the shift invariance \eqref{eq:SI} does not hold in general, we show there exists a subset of EdgeGFs that satisfy this property. Specifically, by considering the filter as a matrix [cf. (1)], the condition \eqref{eq:SI} is equivalent to that the eigenvectors of $\bbS$ and $\bbH_{\rm Edge}(\cdot, \bbS)$ coincide. Following \cite{Coutino2019}, we consider the fixed-support matrices that are diagonalizable w.r.t. the eigenvectors $\bbV$ of $\bbS$, i.e., 
\begin{align}\label{eq:fixedSupport}
	\Omega_\bbV^{\ccalA} \!=\! \{ \bbPhi: \bbPhi \!=\! \bbV \bbphi \bbV^{-1}, [{\rm vec}(\bbPhi)]_i \!=\! 0,~\text{for}~i \!\in\! \ccalA \},
\end{align}
where $\ccalA$ is the index set defining the zero entries of $\bbS + \bbI$, $\bbphi$ is diagonal, and ${\rm vec(\cdot)}$ is the vectorization operation. Let $\bbA \in \{0,1\}^{|\ccalA|\times N^2}$ be a selection matrix, whose rows are the rows of an $N^2 \times N^2$ identity matrix indexed by $\ccalA$. The condition in $\Omega_\bbV^{\ccalA}$ is equivalent to 
\begin{align}
	\bbA {\rm vec}(\bbPhi) = \bbA (\bbV^{-\rm T} * \bbV ) \bbomega = \bb0,
\end{align}
where $*$ is the Kathri-Rao product and $\bbomega = [[\bbphi]_{11},...[\bbphi]_{nn}]^\top$ is the eigenvalue vector. This indicates that $\bbomega$ lies in the nullspace of 
$\bbA (\bbV^{-\rm T} * \bbV )$, i.e., $\bbomega \in {\rm null}\{\bbA (\bbV^{-\rm T} * \bbV )\}$. Let $\bbphi_{\bbV}^\ccalA$ be a basis of $\bbA (\bbV^{-\rm T} * \bbV )$. We can represent the matrix set in \eqref{eq:fixedSupport} as
\begin{align}\label{eq:fixedSupport2}
	\Omega_\bbV^{\ccalA} \!=\! \{ \bbPhi\!:\! \bbPhi \!=\! \bbV \text{diag}(\bbphi_{\bbV}^\ccalA \bbalpha) \bbV^{-1}\!,~\text{for}~\bbalpha \!\in\! \mathbb{R}^p\},
\end{align}
where $p$ is the dimension of ${\rm null}\{\bbA (\bbV^{-\rm T} * \bbV )\}$ and $\bbalpha$ is the expansion coefficients of $\bbphi_{\bbV}^\ccalA$. Therefore, $\Omega_\bbV^{\ccalA}$ is the set of matrices that share the same support as $\bbS + \bbI$ and the same eigenvectors as $\bbS$, and the SI-EdgeGF can be difined by using $\Omega_\bbV^{\ccalA}$. The dimension $p$ determines the DoFs of the parameter set $\Omega_\bbV^{\ccalA}$ and thus, the size of the subclass SI-EdgeGF w.r.t. $\bbV$.

\section{Proof of Theorem \ref{theorem:SIEVfilterStability}} \label{proof:thm1}

The output difference of the SI-EdgeGF over the underlying graph $\bbS$ and the perturbed graph $\widetilde{\bbS}$ is given by\footnote{Throughout the proofs, we assume $\sum_{a}^b = 0$ if $b < a$.}
\begin{align}\label{eq:proofThm11}
	\bbH_{\rm SI}(\bbx, \widetilde{\bbS}) - \bbH_{\rm SI}(\bbx, \bbS) = \sum_{k=0}^K \bbPhi^{(k)} (\widetilde{\bbS}^k \bbx - \bbS^k\bbx).
\end{align}
By representing $\widetilde{\bbS} = \bbS + \bbE\bbS + \bbS\bbE$ with $\bbE$ the relative perturbation matrix, we have
\begin{align}\label{eq:proofThm12}
	\widetilde{\bbS}^k &= (\bbS + \bbE\bbS + \bbS\bbE)^k\\
	&= \bbS^k + \sum_{r=0}^{k-1} \bbS^r \bbE \bbS^{k-r} + \sum_{r=0}^{k-1} \bbS^{r+1} \bbE \bbS^{k-r-1} + \bbD, \nonumber
\end{align}
where $\bbD$ is the sum of the other expanding terms. Substituting \eqref{eq:proofThm12} into \eqref{eq:proofThm11} yields
\begin{align}\label{eq:proofThm13}
	&\bbH_{\rm SI}(\bbx, \widetilde{\bbS}) - \bbH_{\rm SI}(\bbx, \bbS) \\
	&= \sum_{k=1}^K \bbPhi^{(k)} \!\sum_{r=0}^{k-1} (\bbS^r \bbE \bbS^{k-r} \!+\! \bbS^{r+1} \bbE \bbS^{k-r-1}) \bbx \!+\! \sum_{k=1}^K \bbPhi^{(k)} \bbD \bbx. \nonumber
\end{align}    
We consider three terms in \eqref{eq:proofThm13} separately.

\textbf{First two terms.} We expand the graph signal $\bbx = \sum_{i=1}^n \hat{x}_i \bbv_i$ with the eigenvectors $\{\bbv_i\}_{i=1}^n$ of $\bbS$. By substituting this expansion into the first two terms in \eqref{eq:proofThm13}, we get
\begin{align}\label{eq:proofThm14}
	&\sum_{k=1}^K \bbPhi^{(k)} \sum_{r=0}^{k-1} (\bbS^r \bbE \bbS^{k-r} \bbx + \bbS^{r+1} \bbE \bbS^{k-r-1}) \bbx \nonumber \\
	&= \sum_{i=1}^n \hat{x}_i \sum_{k=1}^K \bbPhi^{(k)} \sum_{r=0}^{k-1} \big(\lambda_i^{k-r} \bbS^r + \lambda_i^{k-r-1} \bbS^{r+1}\big) \bbE \bbv_i,
\end{align}
where the fact $\bbS \bbv_i = \lambda_i \bbv_i$ is used. For the term $\sum_{k=1}^K \bbPhi^{(k)} \sum_{r=0}^{k-1} \big(\lambda_i^{k-r} \bbS^r + \lambda_i^{k-r-1} \bbS^{r+1}\big)$ in \eqref{eq:proofThm14} and any graph signal $\bba = \sum_{j=1}^n \hat{a}_j \bbv_j$, we have
\begin{align}\label{eq:proofThm15}
	&\sum_{k=1}^K \bbPhi^{(k)} \sum_{r=0}^{k-1} \big(\lambda_i^{k-r} \bbS^r + \lambda_i^{k-r-1} \bbS^{r+1}\big) \bba \\ 
	& = \sum_{j=1}^n \hat{a}_j \sum_{k=1}^K \phi_j^{(k)} \sum_{r=0}^{k-1} \big(\lambda_i^{k-r} \lambda_j^r + \lambda_i^{k-r-1} \lambda_j^{r+1}\big) \bbv_j \nonumber
\end{align}
where the fact $\bbPhi^{(k)} \bbv_j = \phi^{(k)}_j\bbv_j$ is used because $\bbPhi^{(k)}$ shares the same eigenvectors as $\bbS$ and $\phi_{j}^{(k)}$ is the $j$th eigenvalue of the edge-weight matrix $\bbPhi^{(k)}$. We now consider the term $\sum_{k=1}^K \phi_{j}^{(k)} \sum_{r=0}^{k-1} \big(\lambda_i^{k-r} \lambda_j^r + \lambda_i^{k-r-1} \lambda_j^{r+1}\big) \bbv_{j}$ in \eqref{eq:proofThm15}. Given the filter frequency response function $h_j(\lambda)$ as 
\begin{align}\label{eq:proofThm16}
	h_{j}(\lambda) = \sum_{k=0}^K \phi_{j}^{(k)} \lambda^k,~\forall~ j=1,\ldots,n,
\end{align}
we have
\begin{align}\label{eq:proofThm17}
	&\sum_{k=1}^K \phi_{j}^{(k)} \sum_{r=0}^{k-1} \big(\lambda_i^{k-r} \lambda_j^r + \lambda_i^{k-r-1} \lambda_j^{r+1}\big) \\
	&=
	\begin{cases}
		2 \lambda_j h_{j}'(\lambda_{j}), & \mbox{if}~i = j, \\
		(\lambda_i + \lambda_j)\frac{h_{j}(\lambda_i) - h_{j}(\lambda_{j})}{\lambda_i - \lambda_{j}}, & \mbox{if}~i \ne j, \nonumber
	\end{cases}
\end{align}
where $h_{j}'(\cdot)$ is the derivative of $h_j(\cdot)$. By using the integral Lipschitz condition in \eqref{eq:proofThm17}, we get 
\begin{align}\label{eq:proofThm18}
	\Big| \sum_{k=1}^K \phi_{j}^{(k)} \sum_{r=0}^{k-1} \big(\lambda_i^{k-r} \lambda_j^r + \lambda_i^{k-r-1} \lambda_j^{r+1}\big) \Big| \le 2C_{L}.
\end{align}    
By substituting this result into \eqref{eq:proofThm15}, we have
\begin{align}\label{eq:proofThm19}
	&\Big\|\sum_{k=1}^K \bbPhi^{(k)} \sum_{r=0}^{k-1} \big(\lambda_i^{k-r} \bbS^r + \lambda_i^{k-r-1} \bbS^{r+1}\big) \bba\Big\|_2^2 \\
	&=\! \sum_{j=1}^n\! \hat{a}_j^2 \Big|\sum_{k=1}^K \phi_{j}^{(k)} \sum_{r=0}^{k-1} \big(\lambda_i^{k-r} \lambda_j^r + \lambda_i^{k-r-1} \lambda_j^{r+1}\big)\Big|^2 \nonumber \\
	&\le 4 C_L^2 \sum_{j=1}^n \hat{a}_i^2 = 4 C_L^2 \| \bba \|_2^2, \nonumber
\end{align}
where $\bbv^\top \bbv = 1$ is used in the first equality and $\| \bba \|_2^2 = \sum_{j=1}^n \hat{a}_j^2$ are used in the last equality. 
Therefore, we have
\begin{align}\label{eq:proofThm120}
	\Big\|\sum_{k=0}^K \bbPhi^{(k)} \sum_{r=0}^{k-1} \big(\lambda_i^{k-r} \bbS^r + \lambda_i^{k-r-1} \bbS^{r+1}\big) \Big\|_2 \le 2 C_L.
\end{align}
With the triangular inequality, we bound the first two terms in \eqref{eq:proofThm13} as
\begin{align}\label{eq:proofThm121}
	&\Big\|\sum_{k=1}^K \bbPhi^{(k)} \sum_{r=0}^{k-1} \big(\bbS^r \bbE \bbS^{k-r} + \bbS^{r+1} \bbE \bbS^{k-r-1}\big) \bbx\Big\|_2 \\
	&\le\! \sum_{i=1}^n \!|\hat{x}_i| \Big\|\!\sum_{k=1}^K\!\! \bbPhi^{(k)}\! \sum_{r=0}^{k-1}\!\! \big(\lambda_i^{k-r} \bbS^r \!+\! \lambda_i^{k\!-\!r\!-\!1} \bbS^{r\!+\!1}\big)\!\Big\|_2 \|\bbE\|_2 \|\bbv_i\|_2.\nonumber
\end{align}
By substituting \eqref{eq:proofThm120} and the facts $\|\bbE \|_2 \le \eps$, $\|\bbv_i\|_2 = 1$, $\sum_{i=1}^n |\hat{x}_i| = \|\hat{\bbx}\|_1 \le \sqrt{n}\|\hat{\bbx}\|_2 = \sqrt{n}\|\bbx\|_2$ into \eqref{eq:proofThm121}, we obtain
\begin{align}\label{eq:proofThm122}
	&\Big\|\!\sum_{k=1}^K\!\! \bbPhi^{(k)}\! \sum_{r=0}^{k-1}\!(\bbS^r \bbE \bbS^{k\!-\!r} \!\!+\! \bbS^{r\!+\!1} \bbE \bbS^{k\!-\!r\!-\!1}) \bbx\Big\|_2 \!\!\le\! 2 \sqrt{n} C_L \eps \|\bbx\|_2. 
\end{align}

\textbf{Third term.} Since $\bbD$ consists of the expanding terms that contain at least two error matrices $\bbE$ in the multiplication, it holds that
\begin{align}\label{eq:proofThm123}
	\| \bbD \bbx \|_2 \le \| \bbD \|_2 \|\bbx\|_2 \le \ccalO(\eps^2) \|\bbx\|_2.
\end{align}

By substituting \eqref{eq:proofThm122} and \eqref{eq:proofThm123} into \eqref{eq:proofThm13}, we complete the proof
\begin{align}\label{eq:proofThm235}
	\|\bbH_{\rm SI}(\bbx, \widetilde{\bbS}) \!-\! \bbH_{\rm SI}(\bbx, \bbS)\|_2 \!\le\! 2 \sqrt{n} C_L \| \bbx\|_2 \eps \!+\! \ccalO(\eps^2).
\end{align}


\section{Proof of Theorem \ref{theorem:EI-EVfilterStability}} \label{proof:theorem2}

Following \eqref{eq:proofThm11}-\eqref{eq:proofThm13} in the proof of Theorem 1, we have
\begin{align}\label{eq:proofThm21}
	&\bbH_{\rm ES}(\bbx, \widetilde{\bbS}) - \bbH_{\rm ES}(\bbx, \bbS) \\
	&= \sum_{k=1}^K \bbPhi^{(k)}\! \sum_{r=0}^{k-1} \big(\bbS^r \bbE \bbS^{k-r} \!\!+\! \bbS^{r+1} \!\bbE \bbS^{k-r-1}\big) \bbx \!+\! \sum_{k=1}^K \bbPhi^{(k)} \bbD \bbx, \nonumber
\end{align}    
and we consider three terms in \eqref{eq:proofThm21} separately.

\textbf{First two terms.} We can similarly represent the first two terms in \eqref{eq:proofThm21} as [cf. \eqref{eq:proofThm14}] 
\begin{align}\label{eq:proofThm22}
	&\sum_{k=1}^K \bbPhi^{(k)} \sum_{r=0}^{k-1} \big(\bbS^r \bbE \bbS^{k-r} \!+ \bbS^{r+1} \bbE \bbS^{k-r-1}\big) \bbx \\
	&= \sum_{i=1}^n \hat{x}_i \sum_{k=1}^K \bbPhi^{(k)} \sum_{r=0}^{k-1} \big(\lambda_i^{k-r} \bbS^r \!+\! \lambda_i^{k-r-1} \bbS^{r+1}\big) \bbE \bbv_i. \nonumber
\end{align}
For the term $\sum_{k=1}^K \bbPhi^{(k)} \sum_{r=0}^{k-1} \big(\lambda_i^{k-r} \bbS^r \!+\! \lambda_i^{k-r-1} \bbS^{r+1}\big)$ in \eqref{eq:proofThm22} and any graph signal $\bba = \sum_{j=1}^n \hat{a}_j \bbv_j$, we have
\begin{align}\label{eq:proofThm23}
	&\sum_{k=1}^K\! \bbPhi^{(k)}\! \sum_{r=0}^{k-1}\! \big(\lambda_i^{k-r} \bbS^r \!+\! \lambda_i^{k-r-1} \bbS^{r+1}\big) \bba \\
	&=\! \sum_{j=1}^n \hat{a}_j \sum_{k=1}^K \bbPhi^{(k)} \sum_{r=0}^{k-1} \big(\lambda_i^{k-r} \lambda_j^r + \lambda_i^{k-r-1} \lambda_j^{r+1}\big) \bbv_j, \nonumber
\end{align}
where the fact $\bbS^r \bbv_j = \lambda_j^r \bbv_j$ is used. Let $\bbPhi^{(k)} = \bbU \bbphi^{(k)} \bbU^\top$ be the eigendecomposition of the edge weight matrix with orthogonal eigenvectors $\bbU = [\bbu_1,\ldots,\bbu_n]$ and eigenvalues $\bbphi^{(k)} = {\rm diag}(\phi_1^{(k)},\ldots,\phi_n^{(k)})$. By regarding each eigenvector $\bbv_j$ of $\bbS$ as a graph signal and expanding the latter over $\bbU$ as $\bbv_j = \sum_{j_1=1}^n \hat{v}_{jj_1}\bbu_{j_1}$, we have
\begin{align}\label{eq:proofThm24}
	\bbPhi^{(k)}\bbv_j = \sum_{j_1=1}^n \hat{v}_{jj_1} \bbPhi^{(k)}\bbu_{j_1} = \sum_{j_1=1}^n \hat{v}_{jj_1}\phi^{(k)}_{j_1} \bbu_{j_1},
\end{align}
where the fact $\bbPhi^{(k)}\bbu_{j_1} = \phi_{j_1}^{(k)} \bbu_{j_1}$ is used. 
By leveraging \eqref{eq:proofThm24} in \eqref{eq:proofThm23}, we get
\begin{align}\label{eq:proofThm25}
	&\hat{a}_j \sum_{k=1}^K \bbPhi^{(k)} \sum_{r=0}^{k-1}\big(\lambda_i^{k-r} \lambda_j^r + \lambda_i^{k-r-1} \lambda_j^{r+1}\big) \bbv_j \\
	&= \hat{a}_j  \sum_{j_1=1}^n \hat{v}_{jj_1}\sum_{k=1}^K\phi^{(k)}_{j_1} \sum_{r=0}^{k-1} \big(\lambda_i^{k-r} \lambda_j^r \!+\! \lambda_i^{k-r-1} \lambda_j^{r+1}\big) \bbu_{j_1}.  \nonumber
\end{align}   
For the term $\sum_{k=1}^K \phi_{j_1}^{(k)} \sum_{r=0}^{k-1} \big(\lambda_i^{k-r} \lambda_j^r + \lambda_i^{k-r-1} \lambda_j^{r+1}\big)$ in \eqref{eq:proofThm25}, by using \eqref{eq:proofThm17} in the proof of Theorem 1 with the integral Lipschitz condition, we have
\begin{align}\label{eq:proofThm28}
	\Big| \sum_{k=1}^K \phi_{j_1}^{(k)} \sum_{r=0}^{k-1} \big(\lambda_i^{k-r} \lambda_j^r + \lambda_i^{k-r-1} \lambda_j^{r+1}\big) \Big| \le 2 C_L.
\end{align}    
Leveraging \eqref{eq:proofThm25} and the triangular inequality in \eqref{eq:proofThm23} yields
\begin{align}\label{eq:proofThm29}
	&\Big\|\sum_{j=1}^n\hat{a}_j \sum_{k=1}^K \bbPhi^{(k)} \sum_{r=0}^{k-1} \big(\lambda_i^{k-r} \lambda_j^r + \lambda_i^{k-r-1} \lambda_j^{r+1}\big) \bbv_j \Big\|_2 \\
	&\le \Big\|\sum_{j=1}^n \hat{a}_j \hat{v}_{jj}\sum_{k=1}^K\phi^{(k)}_{j} \sum_{r=0}^{k-1} \big(\lambda_i^{k-r} \lambda_j^r + \lambda_i^{k-r-1} \lambda_j^{r+1}\big) \bbu_{j}\Big\|_2 \nonumber \\
	& +\! \Big\|\!\sum_{j=1}^n\! \hat{a}_j  \!\!\sum_{j_1\ne j}\! \hat{v}_{jj_1}\!\sum_{k=1}^K\!\phi^{(k)}_{j_1}\! \sum_{r=0}^{k-1}\! \big(\lambda_i^{k-r} \lambda_j^r \!+\! \lambda_i^{k-r-1} \lambda_j^{r+1}\big) \bbu_{j_1}\!\Big\|_2\!.  \nonumber
\end{align}  
We handle two norms in \eqref{eq:proofThm29} separately. For the first norm, we have 
\begin{align}\label{eq:proofThm210}
	&\Big\|\sum_{j=1}^n \hat{a}_j \hat{v}_{jj}\!\sum_{k=1}^K\phi^{(k)}_{j} \!\sum_{r=0}^{k-1} \big(\lambda_i^{k-r} \lambda_j^r \!+\! \lambda_i^{k-r-1} \lambda_j^{r+1}\big) \bbu_{j}\Big\|_2^2 \\
	& = \sum_{j=1}^n\!\Big|\hat{a}_j \hat{v}_{jj} \sum_{k=1}^K\phi^{(k)}_{j} \sum_{r=0}^{k-1} \big(\lambda_i^{k-r} \lambda_j^r \!+\! \lambda_i^{k-r-1} \lambda_j^{r+1}\big)\Big|^2 \| \bbu_{j} \|_2^2 \nonumber \\
	& \le\! \sum_{j=1}^n\!|\hat{a}_j |^2 \Big|\sum_{k=1}^K\phi^{(k)}_{j} \sum_{r=0}^{k-1} \big(\lambda_i^{k-r} \lambda_j^r + \lambda_i^{k-r-1} \lambda_j^{r+1}\big)\Big|^2 \| \bbu_{j} \|_2^2 \nonumber \\
	& \le 4 C_L^2 \| \bba \|_2^2, \nonumber
\end{align}    
where the facts $|\hat{v}_{jj}| \le 1$ is used in the first inequality and $\| \bbu_{j} \|_2^2 = 1$, \eqref{eq:proofThm28} are used in the second inequality. For the second norm in \eqref{eq:proofThm29}, we have 
\begin{align}\label{eq:proofThm211}
	&\Big\|\sum_{j=1}^n \hat{a}_j \sum_{j_1\ne j} \hat{v}_{jj_1}\sum_{k=1}^K\phi^{(k)}_{j_1} \sum_{r=0}^{k-1} \big(\lambda_i^{k-r} \lambda_j^r \!+\! \lambda_i^{k-r-1} \lambda_j^{r+1}\big) \bbu_{j_1}\Big\|_2^2 \nonumber \\
	& = \sum_{j_1=1}^n\!\! \Big|\!\sum_{j \ne j_1}\! \hat{a}_j \hat{v}_{jj_1}\! \sum_{k=1}^K\!\phi^{(k)}_{j_1}\! \sum_{r=0}^{k-1}\! \big(\lambda_i^{k-r} \lambda_j^r \!+\! \lambda_i^{k\!-\!r-1} \lambda_j^{r+1}\big) \!\Big|^2 \!\!\|\bbu_{j_1}\|_2^2 \nonumber \\
	&\le 4 n (n-1) C_L^2 \varepsilon^2 \| \bba \|_2^2,
\end{align}  
where the $\varepsilon$-misalignment that $|\hat{v}_{jj_1}|\le \varepsilon$, the triangular inequality, $\sum_{j \ne j_1} \hat{a}_j^2 \le \|\bba\|_2^2$ and \eqref{eq:proofThm28} are used in the last inequality. By substituting \eqref{eq:proofThm210} and \eqref{eq:proofThm211} into \eqref{eq:proofThm29} and the latter into \eqref{eq:proofThm23}, we have
\begin{align}\label{eq:proofThm212}
	&\Big\|\!\sum_{k=1}^K\!\! \bbPhi^{(k)}\! \sum_{r=0}^{k-1}\!\! \big(\!\lambda_i^{k-r} \bbS^r \!+\! \lambda_i^{k-r-1} \bbS^{r+1}\!\big) \bba \Big\|_2 \!\!\le\! 2 (1 \!+\! n \varepsilon) C_L \|\bba\|_2. 
\end{align}  
With the triangular inequality, we bound the first two terms in \eqref{eq:proofThm21} as
\begin{align}\label{eq:proofThm213}
	&\big\|\sum_{k=1}^K \bbPhi^{(k)} \sum_{r=0}^{k-1} \big(\bbS^r \bbE \bbS^{k-r} + \bbS^{r+1} \bbE \bbS^{k-r-1}\big) \bbx\big\|_2 \\
	&\le\! 
	\sum_{i=1}^n\! |\hat{x}_i| \Big\|\!\sum_{k=1}^K\! \bbPhi^{(k)}\!\! \sum_{r=0}^{k-1}\! \big(\lambda_i^{k-r} \bbS^r \!+\! \lambda_i^{k-r-1} \bbS^{r+1}\big)\!\Big\|_2 \|\bbE\|_2 \|\bbv_i\|_2. \nonumber
\end{align}
By using \eqref{eq:proofThm212} and the facts $\|\bbE \|_2 \le \eps$, $\|\bbv_i\|_2 = 1$, $\sum_{i=1}^n |\hat{x}_i| = \|\hat{\bbx}\|_1 \le \sqrt{n}\|\hat{\bbx}\|_2 = \sqrt{n}\|\bbx\|_2$ in \eqref{eq:proofThm213}, we get 
\begin{align}\label{eq:proofThm214}
	&\big\|\sum_{k=1}^K \!\bbPhi^{(k)}\! \sum_{r=0}^{k-1}\! \big(\bbS^r \bbE \bbS^{k-r} + \bbS^{r+1} \bbE \bbS^{k-r-1}\big) \bbx\big\|_2 \\
	&\le 2\sqrt{n} (1 + n \varepsilon) C_L \|\bbx\|_2 \eps. \nonumber
\end{align}

\textbf{Third term.} Following \eqref{eq:proofThm123}, it holds that 
\begin{align}\label{eq:proofThm215}
	\| \bbD \bbx \|_2 \le \| \bbD \|_2 \|\bbx\|_2 \le \ccalO(\eps^2) \|\bbx\|_2.
\end{align}

By using \eqref{eq:proofThm214} and \eqref{eq:proofThm215} in \eqref{eq:proofThm21}, we complete the proof 
\begin{align}
	\|\bbH_{\rm ES}(\bbx, \widetilde{\bbS}) \!-\! \bbH_{\rm ES}(\bbx, \bbS)\|_2 \!\le\! 2\sqrt{n}(1\!+\!n\varepsilon) C_L \|\bbx\|_2\eps \!+\! \ccalO(\eps^2). 
\end{align}   


\section{Proof of Theorem \ref{theorem:EVfilterStability}}\label{proof:theorem3}

We follow \eqref{eq:proofThm11}-\eqref{eq:proofThm13} in the proof of Theorem 1 and represent the output difference of the EdgeGF as
\begin{align}\label{eq:proofThm31}
	&\bbH_{\rm Edge}(\bbx, \widetilde{\bbS}) - \bbH_{\rm Edge}(\bbx, \bbS) \\
	&= \!\sum_{k=1}^K \bbPhi^{(k)} \sum_{r=0}^{k-1} (\bbS^r \bbE \bbS^{k-r} \!+\! \bbS^{r+1} \bbE \bbS^{k-r-1}) \bbx \!+\! \sum_{k=1}^K \bbPhi^{(k)} \bbD \bbx. \nonumber
\end{align} 
We consider three terms in \eqref{eq:proofThm31} separately.

\textbf{First two terms.} We can similarly represent the first two terms in \eqref{eq:proofThm31} as [cf. \eqref{eq:proofThm14}] 
\begin{align}\label{eq:proofThm32}
	&\sum_{k=1}^K \bbPhi^{(k)} \sum_{r=0}^{k-1} (\bbS^r \bbE \bbS^{k-r} + \bbS^{r+1} \bbE \bbS^{k-r-1}) \bbx \\
	&=\! \sum_{i=1}^n \hat{x}_i \sum_{k=1}^K \bbPhi^{(k)} \sum_{r=0}^{k-1} \big(\lambda_i^{k-r} \bbS^r + \lambda_i^{k-r-1} \bbS^{r+1}\big) \bbE \bbv_i. \nonumber
\end{align}
For the term $\sum_{k=1}^K \bbPhi^{(k)} \sum_{r=0}^{k-1} \big(\lambda_i^{k-r} \bbS^r + \lambda_i^{k-r-1} \bbS^{r+1}\big)$ in \eqref{eq:proofThm32} and any graph signal $\bba = \sum_{j=1}^n \hat{a}_j \bbv_j$, we have
\begin{align}\label{eq:proofThm33}
	&\sum_{k=1}^K \bbPhi^{(k)} \sum_{r=0}^{k-1} \big(\lambda_i^{k-r} \bbS^r + \lambda_i^{k-r-1} \bbS^{r+1}\big) \bba \\ 
	& = \sum_{j=1}^n \hat{a}_j \sum_{k=1}^K \bbPhi^{(k)} \sum_{r=0}^{k-1} \big(\lambda_i^{k-r} \lambda_j^r + \lambda_i^{k-r-1} \lambda_j^{r+1}\big) \bbv_j. \nonumber
\end{align}
Let $\bbPhi^{(k)} = \bbU^{(k)} \bbphi^{(k)} (\bbU^{(k)})^\top$ be the eigendecomposition of the edge weight matrix with eigenvectors $\bbU^{(k)} = [\bbu^{(k)}_1,\ldots,\bbu^{(k)}_n]$ and eigenvalues $\bbphi^{(k)} = {\rm diag}(\phi_1^{(k)},\ldots,\phi_n^{(k)})$. By regarding each eigenvector $\bbv_j$ of $\bbS$ as a new graph signal, we can expand it over $\bbU^{(k)}$ as $\bbPhi^{(k)}\bbv_j = \sum_{j_1=1}^n \hat{v}_{jj_1}^{(k)} \bbPhi^{(k)}\bbu^{(k)}_{j_1}$ where $\{\hat{v}_{jj_1}^{(k)}\}_{j_1=1}^n$ are expanding coefficients. By further expanding each $\bbu^{(k)}_{j_1}$ of $\bbPhi^{(k)}$ over $\bbV$ as $\bbPhi^{(k)}\bbv_j = \sum_{j_2=1}^n\sum_{j_1=1}^n \hat{v}_{jj_1}^{(k)}\hat{u}^{(k)}_{j_1 j_2} \phi^{(k)}_{j_1} \bbv_{j_2}$ with expanding coefficients $\{\hat{u}^{(k)}_{j_1 j_2}\}_{j_2=1}^n$ and using these expansions, we can represent \eqref{eq:proofThm33} as
\begin{align}\label{eq:proofThm35}
	&\hat{a}_j \sum_{k=1}^K \bbPhi^{(k)} \sum_{r=0}^{k-1} \big(\lambda_i^{k-r} \lambda_j^r + \lambda_i^{k-r-1} \lambda_j^{r+1}\big) \bbv_j \\ 
	& = \hat{a}_j \sum_{j_2\!=\!1}^n\sum_{j_1\!=\!1}^n\sum_{k\!=\!1}^K \hat{v}_{jj_1}^{(k)}\hat{u}^{(k)}_{j_1 j_2}\phi^{(k)}_{j_1} \sum_{r=0}^{k-1}\! \big(\lambda_i^{k\!-\!r} \lambda_j^r \!+\! \lambda_i^{k\!-\!r\!-\!1} \lambda_j^{r\!+\!1}\big) \bbv_{j_2} \nonumber \\
	&=\! \hat{a}_j\!\!  \sum_{j_2\!=\!1}^n\!\sum_{j_1\!=\!1}^n\!\! \hat{v}_{jj_1}^{(1)}\hat{u}^{(1)}_{j_1 j_2}\! \sum_{k\!=\!1}^K \!\phi^{(k)}_{j_1}\!c^{(k)}_{jj_1j_2}\! \sum_{r\!=\!0}^{k\!-\!1}\!\! \big(\lambda_i^{k\!-\!r}\! \lambda_j^r \!\!+\!\! \lambda_i^{k\!-\!r\!-\!1}\!\! \lambda_j^{r\!+\!1}\big) \bbv_{j_2}\!,\nonumber
\end{align} 
where $c^{(k)}_{jj_1j_2} = \hat{v}_{jj_1}^{(k)}\hat{u}^{(k)}_{j_1 j_2}/\hat{v}_{jj_1}^{(1)}\hat{u}^{(1)}_{j_1 j_2}$ is the concise notation and $c^{(0)}_{jj_1j_2} = 1$ by default. We consider the term 
\begin{align}\label{eq:proofThm36}
	\sum_{k=1}^K \phi^{(k)}_{j_1}c^{(k)}_{jj_1j_2} \sum_{r=0}^{k-1} \big(\lambda_i^{k-r} \lambda_j^r + \lambda_i^{k-r-1} \lambda_j^{r+1}\big).
\end{align}
For $K$ terms $\big\{c_{jj_1j_2}^{(k)} \sum_{r=0}^{k-1} \big(\lambda_i^{k-r} \lambda_j^r + \lambda_i^{k-r-1} \lambda_j^{r+1}\big)\big\}_{k=1}^K$ in \eqref{eq:proofThm36}, define two multivariate variables $\bblambda_{1, jj_1j_2} = [\lambda_{1,jj_1j_2}^{(1)}, \ldots, \lambda_{1,jj_1j_2}^{(K)}]^\top$ and $\bblambda_{2, jj_1j_2} = [\lambda_{2,jj_1j_2}^{(1)}, \ldots, \lambda_{2,jj_1j_2}^{(K)}]^\top$, where the $k$th entries are scaled versions of the graph frequency
\begin{align}\label{eq:proofThm37}
	\lambda_{1,jj_1j_2}^{(k)} = \frac{c^{(k)}_{jj_1j_2} \lambda_i}{c^{(k-1)}_{jj_1j_2}},~ \lambda_{2,jj_1j_2}^{(k)} = \frac{c^{(k)}_{jj_1j_2} \lambda_j}{c^{(k-1)}_{jj_1j_2}}
\end{align}
for $k=1,...,K$ with $\lambda_{1,jj_1j_2}^{(0)} \!=\! c^{(0)}_{jj_1j_2} \lambda_i$ and $\lambda_{2,jj_1j_2}^{(0)} \!=\! c^{(0)}_{jj_1j_2} \lambda_j$. We can represent \eqref{eq:proofThm36} with the Lipschitz gradient of the edge-varying filter frequency response 
$\nabla h_{j_1}(\bblambda)|_{\bblambda_{1, jj_1j_2}, \bblambda_{2,jj_1j_2}}$, i.e.,
\begin{align}\label{eq:proofThm38}
	&\nabla h_{j_1, jj_1j_2}\!:=\! (\bblambda_{1, jj_1j_2} + \bblambda_{2, jj_1j_2}) \cdot \nabla h_{j_1}(\bblambda)|_{\bblambda_{1, jj_1j_2}, \bblambda_{2,jj_1j_2}} \nonumber \\
	&= \sum_{k=1}^K \phi^{(k)}_{j_1}c^{(k)}_{jj_1j_2} \sum_{r=0}^{k-1} \big(\lambda_i^{k-r} \lambda_j^r + \lambda_i^{k-r-1} \lambda_j^{r+1}\big),
\end{align}
where $\nabla h_{j_1, jj_1j_2}$ is the concise notation and $\cdot$ is the inner product of vectors [Def. 6]. By substituting \eqref{eq:proofThm38} into \eqref{eq:proofThm35}, we have
\begin{align}\label{eq:proofThm39}
	&\sum_{j=1}^n \hat{a}_j \sum_{k=1}^K \bbPhi^{(k)} \sum_{r=0}^{k-1} \big(\lambda_i^{k-r} \lambda_j^r + \lambda_i^{k-r-1} \lambda_j^{r+1}\big) \bbv_j \\
	&= \sum_{j=1}^n\sum_{j_2=1}^n\sum_{j_1=1}^n \hat{a}_j \hat{v}_{jj_1}^{(1)}\hat{u}^{(1)}_{j_1 j_2} \nabla h_{j_1, jj_1j_2} \bbv_{j_2}. \nonumber
\end{align}   
From the $\varepsilon$-misalignment that $|\hat{v}_{j_1 j_2}| \le \varepsilon$, $|\hat{u}_{j_1 j_2}| \le \varepsilon$ for all $j_1 \ne j_2$ and the triangular inequality, we have
\begin{align}\label{eq:proofThm310}
	&\Big\|\sum_{j=1}^n\hat{a}_j \sum_{k=0}^K \bbPhi^{(k)} \sum_{r=0}^{k-1} \!\big(\lambda_i^{k-r} \lambda_j^r \!+\! \lambda_i^{k-r-1} \lambda_j^{r+1}\big) \bbv_j\Big\|_2 \\
	&\le \Big\|\sum_{j=1}^n \hat{a}_j \hat{v}_{jj}^{(1)}\hat{u}^{(1)}_{j j} \nabla h_{j, jjj} \bbv_{j}\Big\|_2 \nonumber \\
	&+ \Big\|\sum_{j_1=1}^n \sum_{j=1}^n \hat{a}_j \hat{v}_{jj_1}^{(1)}\hat{u}^{(1)}_{j_1 j_1}\nabla h_{j_1, jj_1j_1} \bbv_{j_1}\Big\|_2\nonumber \\
	&+ \Big\|\sum_{j_1=1}^n \sum_{j=1}^n \hat{a}_j \hat{v}_{jj}^{(1)}\hat{u}^{(1)}_{j j_1}\nabla h_{j, jjj_1} \bbv_{j_1} \Big\|_2 + \ccalO(\varepsilon^2),  \nonumber
\end{align} 
where the last term $\ccalO(\varepsilon^2)$ is the sum of the other expanding terms, which contain two misalignment coefficients $\hat{v}_{jj_1}^{(1)}$, $\hat{u}^{(1)}_{j_1 j_2}$ with $j \ne j_1$, $j_1 \ne j_2$ and are therefore on the order of $\varepsilon^2$ for eigenvectors $\{\bbU^{(k)}\}_{k=1}^K$ close to $\bbV$ with a small $\varepsilon$. We then consider three norms in \eqref{eq:proofThm310} separately. 
For the first norm, since $|\hat{v}_{jj}^{(1)}| \le 1$ and $|\hat{u}^{(1)}_{j j}| \le 1$ with the orthonormal eigendecomposition, we have 
\begin{align}\label{eq:proofThm311}
	&\Big\|\sum_{j=1}^n \hat{a}_j \hat{v}_{jj}^{(1)}\hat{u}^{(1)}_{j j}\nabla h_{j, jjj} \bbv_{j}\Big\|_2^2 \\ 
	&=\sum_{j=1}^n\Big|\hat{a}_j \hat{v}_{jj}^{(1)}\hat{u}^{(1)}_{j j} \nabla h_{j, jjj}\Big|^2 \| \bbv_{j} \|_2^2 \nonumber\\
	&\le \sum_{j=1}^n|\hat{a}_j |^2 \big|\nabla h_{j, jjj}\big|^2 \| \bbv_{j} \|_2^2 
	\le 4 C_L^2 \| \bba \|_2^2, \nonumber 
\end{align}    
where the integral Lipschitz property of $h_j(\bblambda)$ [cf. (33)], i.e., $|\nabla h_{j, jjj}| \le 2 C_L$, and $\|\bbv_j\|_2 = 1$ are used in the last inequality. For the second norm, since $|\hat{v}_{jj_1}^{(1)}| \le \varepsilon$ and $|\hat{u}^{(1)}_{j_1 j_1}| \le 1$, we have 
\begin{align}\label{eq:proofThm312}
	&\Big\|\sum_{j_1=1}^n \sum_{j=1}^n \hat{a}_j \hat{v}_{jj_1}^{(1)}\hat{u}^{(1)}_{j_1 j_1}\nabla h_{j_1, jj_1j_1} \bbv_{j_1}\!\Big\|_2^2 \\
	&=\sum_{j_1\!=\!1}^n\! \Big|\!\sum_{j=1}^n\! \hat{a}_j\! \hat{v}_{jj_1}^{(1)}\!\hat{u}^{(1)}_{j_1 j_1}\nabla h_{j_1, jj_1j_1} \!\Big|^2\! \|\bbv_{j_1}\|_2^2 \!\le\! 4 n^2 C_L^2 \varepsilon^2 \| \bba \|_2^2. \nonumber
\end{align}      
A similar result holds for the third term 
\begin{align}\label{eq:proofThm313}
	&\Big\|\!\sum_{j_1=1}^n \sum_{j=1}^n \hat{a}_j \hat{v}_{jj}^{(1)}\hat{u}^{(1)}_{j j_1} \nabla h_{j, jjj_1} \bbv_{j_1} \Big\|_2^2 \!\le\! 4 n^2 C_L^2 \varepsilon^2 \| \bba \|_2^2.
\end{align}        
By using \eqref{eq:proofThm311}, \eqref{eq:proofThm312} and \eqref{eq:proofThm313} in \eqref{eq:proofThm310}, 
we have
\begin{align}\label{eq:proofThm314}
	&\Big\|\sum_{k=1}^K \bbPhi^{(k)} \sum_{r=0}^{k-1} (\lambda_i^{k-r} \bbS^r + \lambda_i^{k-r-1} \bbS^{r+1})\Big\|_2 \\
	&\le 2 (1 + 2 n \varepsilon) C_L + \ccalO(\varepsilon^2). \nonumber
\end{align}
With the triangular inequality, we bound the first two terms as
\begin{align}\label{eq:proofThm315}
	&\Big\|\sum_{k=1}^K \bbPhi^{(k)} \sum_{r=0}^{k-1} \big(\bbS^r \bbE \bbS^{k-r} + \bbS^{r+1} \bbE \bbS^{k-r-1} \big) \bbx \Big\|_2 \\
	&\le \sum_{i=1}^n\! |\hat{x}_i| \Big\|\!\sum_{k=1}^K\! \bbPhi^{(k)}\! \sum_{r=0}^{k-1}\!\! \big(\!\lambda_i^{k-r} \bbS^r \!\!+\!\! \lambda_i^{k-r-1} \bbS^{r+1}\big)\!\Big\|_2 \|\bbE\|_2 \|\bbv_i\|_2. \nonumber
\end{align}
By using \eqref{eq:proofThm314} and the facts $\|\bbE \|_2 \le \eps$, $\|\bbv_i\|_2 = 1$, $\sum_{i=1}^n |\hat{x}_i| = \|\hat{\bbx}\|_1 \le \sqrt{n}\|\hat{\bbx}\|_2 = \sqrt{n}\|\bbx\|_2$ in \eqref{eq:proofThm315}, we have
\begin{align}\label{eq:proofThm316}
	&\Big\|\sum_{k=1}^K \bbPhi^{(k)} \sum_{r=0}^{k-1} \big(\bbS^r \bbE \bbS^{k-r} + \bbS^{r+1} \bbE \bbS^{k-r-1} \big) \bbx \Big\|_2 \\
	& \le 2\sqrt{n}(1 + 2 n \varepsilon) C_L \|\bbx\|_2 \eps + \ccalO(\varepsilon^2). \nonumber
\end{align}

\textbf{Third term.} Following \eqref{eq:proofThm123}, we have
\begin{align}\label{eq:proofThm317}
	\| \bbD \bbx \|_2 \le \| \bbD \|_2 \|\bbx\|_2 \le \ccalO(\eps^2) \|\bbx\|_2.
\end{align}

By using \eqref{eq:proofThm316} and \eqref{eq:proofThm317} in \eqref{eq:proofThm31}, we complete the proof 
\begin{align}\label{eq:proofThm318}
	&\|\bbH_{\rm Edge}(\bbx, \widetilde{\bbS}) - \bbH_{\rm Edge}(\bbx, \bbS)\|_2 \\
	&\le 2 \sqrt{n}(1 + 2 n \varepsilon) C_L \|\bbx\|_2 \eps + \ccalO(\varepsilon^2) + \ccalO(\eps^2). \nonumber
\end{align}


\section{Proof of Theorem \ref{thm:stabilityEdgeNet}} \label{proof:theorem4}

With the underlying graph $\bbS$ and the perturbed graph $\widetilde{\bbS}$, the output difference of the EdgeNet is given by\footnote{We remove the subscript $\rm Edge$ for convenience of expression.}
\begin{align} \label{eqn:thm41}
	&\|\bbPsi(\!\bbx, \bbS, \ccalH)-\bbPsi(\bbx, \widetilde{\bbS}, \ccalH)\|_2 \\
	&= \Big\| \sigma\Big( \sum_{f=1}^{F} \bbH_{L}^{1f}(\bbx^f_{L-1}, \bbS) \Big) - \sigma\Big( \sum_{f=1}^{F} \bbH_{ L}^{1f}(\widetilde{\bbx}^f_{L-1}, \widetilde{\bbS}) \Big) \Big\|_2 \nonumber \\
	&\le \sum_{f=1}^{F}  \big\| \bbH_{L}^{1f}\!(\bbx^f_{L-1}, \bbS) - \bbH_{L}^{1f}(\widetilde{\bbx}^f_{L-1}, \widetilde{\bbS}) \big\|_2, \nonumber
\end{align}
where $\widetilde{\bbx}_{L-1}^f$ is the $f$th feature of the $(L-1)$th layer output over $\widetilde{\bbS}$ instead of $\bbS$, 
and where the last inequality uses the normalized Lipschitz nonlinearity $\sigma(\cdot)$ and the triangle inequality. By adding and substracting $\bbH_{L}^{1f}(\bbx^f_{L-1}, \widetilde{\bbS})$, the term $\big\| \bbH_{L}^{1f}\!(\bbx^f_{L-1}, \bbS) - \bbH_{L}^{1f}(\widetilde{\bbx}^f_{L-1}, \widetilde{\bbS}) \big\|_2$ in \eqref{eqn:thm41} is bounded by
\begin{align} \label{eqn:thm42}
	&\big\| \bbH_{L}^{1f}\!\!(\bbx^f_{L\!-\!1}\!,\! \bbS) \!-\!\! \bbH_{L}^{1f}\!\!(\bbx^f_{L\!-\!1}\!, \!\widetilde{\bbS}) \big\|_2\!\!+\!\!\big\|\bbH_L^{1f}\!\!(\bbx^f_{L\!-\!1}\!,\! \widetilde{\bbS}) \!\!-\!\! \bbH_L^{1f}\!\!(\widetilde{\bbx}^f_{L\!-\!1}\!, \widetilde{\bbS}) \big\|_2  \nonumber\\
	& \le\! \big\| \bbH_L^{1f}\!(\bbx^f_{L\!-\!1}\!,\! \bbS\!) \!\!-\!\!\bbH_L^{1f}\!(\bbx^f_{L\!-\!1}\!,\! \widetilde{\bbS}) \big\|_2 \!\!+\!\! \| \bbH_L^{1f}\!\!(\cdot,\! \widetilde{\bbS}) \|_2 \big\|\bbx^f_{L\!-\!1} \!\!-\!\! \widetilde{\bbx}^f_{L\!-\!1}\! \big\|_2,
\end{align}
where $\bbH_L^{1f}(\cdot, \widetilde{\bbS})$ is the filter matrix over $\widetilde{\bbS}$. 

For the first term in \eqref{eqn:thm42}, by using Theorem 3, we can bound it as 
\begin{align} \label{eqn:thm43}
	&\big\| \bbH_L^{1f}(\bbx^f_{L-1}, \bbS) -\bbH_L^{1f}(\bbx^f_{L-1}, \widetilde{\bbS}) \big\|_2 \\
	&\le C_{\rm Edge} \|\bbx^f_{L-1}\|_2 \eps + \ccalO(\varepsilon^2) + \ccalO(\eps^2), \nonumber
\end{align}
where $C_{\rm Edge} = 2\sqrt{n}(1 + 2 n \varepsilon) C_L$ is the stability constant. For the term $\|\bbx^f_{L-1}\|_2$, we observe that 
\begin{align} \label{eqn:thm44}
	&\|\bbx^f_{L-1}\|_2 = \Big\| \sigma\Big( \sum_{g=1}^{F} \bbH_{L-1}^{fg}(\bbx^g_{L-2}, \bbS)\Big) \Big\|_2 \\
	&\le \sum_{g=1}^{F} \big\| \bbH_{L-1}^{fg}(\bbx^g_{L-2}, \bbS) \big\|_2 
	\le \sum_{g=1}^{F}\! \big\| \bbx^g_{L-2} \big\|_2, \nonumber 
\end{align}
where we use the triangle inequality, followed by the bound on filters, i.e., $|h^{fg}_{L-1}(\bblambda)| \le 1$. By continuing this recursion, we get 
\begin{align} \label{eqn:thm45}
	\|\bbx^f_{L-1}\|_2 \le F^{L-2} \big\| \bbx^1_0 \big\|_2 = F^{L-2} \| \bbx \|_2
\end{align}
with $\| \bbx^1_0 \|_2 = \| \bbx \|_2$ by definition. By substituting this result into \eqref{eqn:thm43}, we have 
\begin{align} \label{eqn:thm46}
	&\big\| \bbH_L^{1f}(\bbx^f_{L-1}, \bbS) -\bbH_L^{1f}(\bbx^f_{L-1}, \widetilde{\bbS}) \big\|_2 \\
	&\le F^{L-2} C_{\rm Edge} \| \bbx \|_2 \eps + \ccalO(\varepsilon^2) + \ccalO(\eps^2). \nonumber
\end{align}

For the second term in \eqref{eqn:thm42}, using the filter bound again yields
\begin{equation} \label{eqn:thm47}
	\| \bbH_L^{1f}(\cdot, \widetilde{\bbS}) \|_2 \big\|\bbx^f_{L-1} - \widetilde{\bbx}^f_{L-1} \big\|_2 \le \big\|\bbx^f_{L-1} - \widetilde{\bbx}^f_{L-1} \big\|_2.
\end{equation}
By substituting \eqref{eqn:thm46} and \eqref{eqn:thm47} into \eqref{eqn:thm42} and the latter into \eqref{eqn:thm41}, we get
\begin{align} \label{eqn:thm48}
	&\|\bbPsi(\bbx, \bbS, \ccalH)\!-\!\bbPsi(\bbx, \widetilde{\bbS}, \ccalH)\|_2 \!\le\! \sum_{f=1}^{F}\big\|\bbx^f_{L-1} \!-\! \widetilde{\bbx}^f_{L-1} \big\|_2 \nonumber \\
	&+ F^{L-1} C_{\rm Edge} \| \bbx \|_2 \eps + \ccalO(\varepsilon^2) + \ccalO(\eps^2).
\end{align}
We see the recursion that the output difference of the $L$th layer depends on that of the $(L-1)$th layer. By repeating this recursion until the input layer and using the initial condition $\| \bbx_0^1 - \widetilde{\bbx}_0^1 \| = \| \bbx - \bbx \|=0$, we complete the proof
\begin{align} \label{eqn:thm49}
	&\|\bbPsi(\bbx, \bbS, \ccalH)\!-\!\bbPsi(\bbx, \widetilde{\bbS}, \ccalH)\|_2 \\
	&\le L F^{L-1} C_{\rm Edge} \| \bbx \|_2 \eps + \ccalO(\varepsilon^2) + \ccalO(\eps^2). \nonumber
\end{align}


\bibliographystyle{IEEEbib}
\bibliography{myIEEEabrv,biblioOp}

\end{document}